%% file: main_arXiv.tex
\documentclass[11pt]{article}
\input{preamble_arXiv.tex}

\title{Soft GRAND under Channel Uncertainty:\\Minimax Posterior-Envelope Ordering and\\ Gallager-Exponent Preservation}

\author{%
Behrooz~Razeghi\textsuperscript{\orcidlink{0000-0001-9568-4166}}\\[2pt]
\small Harvard University, Cambridge, MA, USA\\
\small \texttt{behroozrazeghi@seas.harvard.edu},
\texttt{behrooz.razeghi@gmail.com}
}

\date{}

\begin{document}

\maketitle

\begin{abstract}
Soft Guessing Random Additive Noise Decoding queries corrections in an observation-dependent order. For a known channel, the matched order ranks corrections by nonincreasing conditional posterior probability. Under channel uncertainty, admissible channels may induce matched orders that assign different ranks to the same correction at the same observation. We instead use a posterior-envelope order, defined by the pointwise supremum of codebook-independent uniform-input correction posteriors. For fixed measurable posterior representatives and a measurable envelope, this order minimizes, at each observation, the worst-case logarithm of the rank--posterior product over admissible channel--correction pairs with positive posterior probability. For uniform-subset codebooks, this finite-block minimax property does not by itself determine the ensemble-average error exponent, which depends on the distribution of the realized correction rank. For i.i.d. correction--observation pairs with a finite correction alphabet whose joint distribution is dominated by the product of counting measure on that alphabet and a $\sigma$-finite measure on the observation space, Arimoto conditional R\'enyi entropy characterizes the matched rank-moment spectrum. An auxiliary posterior-envelope order preserves this spectrum on $[0,1]$ when its posterior family contains the corresponding conditional power tilts and has subexponential conditional Shtarkov complexity on observation sets whose complement probabilities decay exponentially faster than the reciprocal correction-space cardinality. Under additional differentiability and entropy nondegeneracy conditions, the auxiliary order attains the matched uniform-subset ensemble-average error exponent at every fixed normalized code rate strictly between zero and one. For uniform-input memoryless channels, this exponent equals Gallager's uniform-input random-coding exponent.
\end{abstract}


\vspace{15pt}

\section{Introduction}
\label{sec:introduction}

Soft GRAND decodes by querying candidate corrections in an observation-dependent order. For a known channel, any nonincreasing ordering of the conditional correction posterior induces a likelihood ordering of the corresponding candidate codewords. Without abandonment, the first codebook hit is therefore an ML codeword \cite{solomon2020soft}. Under channel uncertainty, two admissible channels may reverse the posterior order of two corrections at the same observation. Their matched query permutations may consequently assign different ranks to the same correction and may return different codewords from the same fixed codebook. In our formulation, the decoder selects one measurable query permutation for each observation without knowing which admissible conditional posterior applies.

We distinguish the finite-block minimax criterion from the fixed-rate ensemble-average error-exponent criterion. At fixed blocklength, we consider worst-case posterior-normalized rank regret over the prescribed family of conditional correction posteriors. For the uniform-subset code ensemble at a fixed normalized code rate strictly between zero and one, the ensemble-average decoding error exponent depends on the distribution of the realized correction rank, including its right tail, through collisions with nontransmitted codewords. Pointwise minimax rank regret therefore does not by itself imply preservation of the matched uniform-subset ensemble-average error exponent.

\vspace{4pt}

Classical universal decoding studies decoding metrics that do not require the unknown channel parameter and compares their random-coding error exponents with those of matched decoding \cite{merhav1993gaussian,merhav2013universal,huleihel2015isi}. For discrete additive channels, Joudeh analyzes a universal GRAND variant that attains the random-coding error exponent without knowledge of the noise distribution \cite{joudeh2024grand}. Miyamoto and Yang construct noise-guessing decoders using noise-sequence probability estimators and give sufficient conditions under which the resulting decoders attain the matched random-coding exponent \cite{miyamoto2025finite,miyamoto2025universal}. Mismatched noise-guessing analyses characterize error and complexity exponents under prescribed sequence metrics, including power-tilted and universal metrics \cite{miyamoto2026mismatched}. ORB-type methods construct observation-dependent orders from reliability statistics under specified channel models \cite{li2026orbfinite,wan2026orbtype}. Other GRAND methods exploit specified correlation or finite-memory posterior structure, or update the query order under channel switching or drift \cite{razeghi2026lpgrand,razeghi2026tail,razeghi2026switchdrift}. The decoder instead selects one correction permutation for each observation while the conditional correction posterior ranges over a family whose members may induce different matched permutations. In contrast to the discrete additive settings of \cite{joudeh2024grand, miyamoto2025finite,miyamoto2025universal}, the observation $R^n$ may be continuous. The observation-wise minimax posterior-rank problem and preservation of the matched random-coding exponent are treated separately, with the latter based on conditional tilt closure and conditional Shtarkov complexity.

\vspace{4pt}

Ar{\i}kan bounds positive moments of conditional guesswork in terms of Arimoto conditional R\'enyi entropy; for i.i.d. source--side-information pairs, these bounds determine the corresponding asymptotic exponential rates \cite{arikan1996}. Li establishes a large-deviation principle for conditional guesswork \cite{li2019conditional}, while Girish \textit{et al.} derive explicit conditional tail exponents for i.i.d.\ sequences with element-wise side information \cite{girish2026conditional}. Guessing under source uncertainty and universal guessing formulations consider uncertainty in the underlying source distribution \cite{sundaresan2007uncertainty,merhavcohen2019}. These results do not formulate the observation-wise minimax selection of a deterministic correction permutation over a family of conditional posteriors considered here. Shtarkov normalized maximum likelihood minimizes worst-case log-loss regret over probability assignments \cite{shtarkov1987}. Liu \textit{et al.}
extend this framework to sequential probability assignment with contexts through contextual Shtarkov sums and contextual normalized maximum likelihood \cite{liu2024contextual}. The finite-block criterion instead optimizes a permutation under posterior-normalized rank regret.

\vspace{4pt}

The prescribed channel class determines the posterior family in the finite-block minimax criterion, whereas the decoder posterior family determines the product-posterior envelope used in the sufficient conditions for exponent preservation. If the posterior family induced by the prescribed channel class satisfies these conditions, the same family may be used as the decoder posterior family. More generally, the decoder posterior family may contain auxiliary posteriors that are not induced by channels in the prescribed class. The sufficient conditions do not establish that this enlargement is necessary.

\clearpage

We establish the following results.

\begin{enumerate}[leftmargin=*]

\item
For each admissible observation, the posterior-envelope order minimizes the worst-case posterior-normalized rank regret over admissible channel--correction pairs with positive posterior probability, and we characterize the exact pointwise minimax value. The corresponding conditional Shtarkov minimax log-loss regret exceeds this value by at most
$\log_a\mathsf H_{a^n}=\mathcal O(\log n)$, where $\mathsf H_N=\sum_{j=1}^N j^{-1}$. For common-support parametric posterior families satisfying a log-Lipschitz condition, finite posterior-grid mixtures bound both the posterior regret and the induced rank regret by the sum of a covering term and a Lipschitz approximation term. Polynomial covering numbers and uniformly bounded approximation terms therefore give $\mathcal O(\log n)$ rank regret. Under a conditional posterior asymptotic equipartition property (AEP), a posterior-universal assignment induces an order that attains the matched first-order query-rank exponent.

\item
For the uniform-subset code ensemble, we derive an exact finite-block rank--collision identity expressing the ensemble-average decoding error probability in terms of the distribution of the realized correction rank. Consequently, convergence in probability of the normalized logarithmic rank alone does not determine the fixed-rate ensemble-average error exponent. For i.i.d. correction--observation pairs with a finite correction alphabet and joint distribution dominated by counting measure on that alphabet times
a $\sigma$-finite measure on a standard Borel observation space, Arimoto conditional R'enyi entropy characterizes the matched positive rank-moment spectrum. When these pairs are induced by a uniform-input memoryless channel, the matched rank-moment spectrum, exposed-point right-tail bounds, and the rank--collision identity yield, under the stated differentiability and nondegeneracy conditions, the matched uniform-subset ensemble-average error exponent for every fixed rate $0<R_{\rm c}<1$ and every code-size sequence $M_n$ satisfying $\log_a M_n=nR_{\rm c}+o(n)$. This exponent equals Gallager's uniform-input random-coding exponent.

\item
We give sufficient conditions under which an auxiliary posterior-envelope order preserves the matched rank-moment spectrum for moment orders $\rho\in[0,1]$. Preservation holds when the decoder posterior family is $1$-tilt closed at the channel parameter and has subexponential conditional Shtarkov complexity on observation sets whose complement probabilities decay exponentially faster than the reciprocal correction-space cardinality. If the matched spectrum is continuously differentiable on $[0,1+\epsilon]$ for some $\epsilon>0$ and the one-letter conditional correction entropy is strictly between zero and one in base-$a$ units, then, for every fixed $0<R_{\rm c}<1$ and every code-size sequence satisfying $\log_a M_n=nR_{\rm c}+o(n)$, the auxiliary order attains the matched uniform-subset ensemble-average error exponent. For a uniform-input memoryless channel, this exponent equals Gallager's uniform-input random-coding exponent. The decoder posterior family may contain auxiliary posteriors that are not induced by channels in the prescribed channel class.

\item
For generalized-Gaussian BPSK with uncertain shape and scale, we verify the posterior-grid and exponent-preservation conditions. At fixed shape, changing the scale leaves the correction order unchanged, whereas changing the shape can reverse the relative order of two correction patterns for the same received block. Conditional power tilting maps $(\beta,\sigma) \longmapsto \left(\beta,\sigma(1+\rho)^{1/\beta}\right)$, so an auxiliary scale range can be chosen to contain all required tilts for $\rho\in[0,1]$. Finite posterior-grid mixtures give $\mathcal O(\log n)$ posterior regret and induced rank regret on suitable observation sets. These observation sets can be chosen with arbitrarily large base-2 exponential decay exponent for their complement probabilities, and hence to satisfy the bad-set condition required for exponent preservation. Separately, the product-posterior envelope over the corresponding tilt-containing decoder posterior family preserves the matched rank-moment spectrum on $[0,1]$ and, under the fixed-rate code-size scaling specified above, attains the matched uniform-subset ensemble-average error exponent.

\end{enumerate}

\section{Related Work}
\label{sec:prior}

For a discrete additive channel with known noise distribution, GRAND queries noise sequences in nonincreasing probability and, without abandonment, returns an ML codeword. For random codebooks, its error and complexity exponents were characterized in \cite{duffy2019capacity}. SGRAND uses soft observations to query candidate corrections in nonincreasing likelihood order for additive memoryless channels and, without abandonment, returns an ML codeword \cite{solomon2020soft}. SRGRAND and DSGRAND incorporate quantized reliability information into the query order \cite{duffy2021guessing,yuan2026discretized}. These methods construct the query order under a specified channel or reliability model. The soft observation is instead retained while the conditional correction posterior ranges over a prescribed family whose members may induce different matched orders at the same observation.

Ar{\i}kan bounds positive moments of optimal guesswork in terms of R\'enyi entropy \cite{arikan1996}, and Christiansen and Duffy establish a large-deviation principle for normalized logarithmic guesswork \cite{christiansen2012guesswork}. Sundaresan studies guessing under source uncertainty with side information and identifies strategies that minimize the supremum guessing redundancy over families of joint source distributions \cite{sundaresan2007uncertainty}. Merhav and Cohen construct universal randomized guessing strategies for unknown memoryless and finite-state sources, including settings with side information, derive their guessing exponents, and establish asymptotic optimality \cite{merhavcohen2019}. Li establishes a large-deviation principle for conditional guesswork \cite{li2019conditional}. For i.i.d. sequences with element-wise side information, Girish \textit{et al.} derive explicit conditional tail exponents using type counting and conditional tilted distributions \cite{girish2026conditional}. For memoryless finite-alphabet sources, Beirami \textit{et al.} show that two source distributions induce the same guessing order if and only if they belong to the same tilted family \cite{beirami2018tilting}. Salamatian \textit{et al.} characterize asymptotic mismatched guesswork using tilted and exponential families \cite{salamatian2019mismatched}, while Mariona \textit{et al.} formulate nonasymptotic mismatch through a variant of Kendall's tau distance between optimal guessing functions \cite{mariona2023mismatch}. These works study guesswork under specified source models or source uncertainty. The present decision problem instead selects, for each soft observation, one deterministic correction permutation against a family of conditional posteriors and evaluates its worst-case posterior-normalized rank regret.

Classical universal decoding constructs decoding metrics that do not require the unknown channel parameter and compares their random-coding error exponents with those of matched decoding over prescribed channel or metric classes \cite{merhav1993gaussian,federlapidoth1998,merhav2013universal,huleihel2015isi}. Universal channel coding for general output alphabets, including continuous
outputs, is studied in \cite{hayashi2018universal}. For memoryless discrete additive channels, Joudeh analyzes a universal GRAND variant that does not require knowledge of the noise distribution and attains the random-coding error exponent \cite{joudeh2024grand}. Miyamoto and Yang use noise-sequence probability estimators to construct deterministic and randomized universal noise-guessing decoders. For parametric discrete additive channels, they give sufficient conditions for random-coding strong universality \cite{miyamoto2025finite}; for finite-state additive channels, they establish random-coding universality \cite{miyamoto2025universal}. Mismatched noise-guessing analyses characterize error and complexity exponents under fixed, $\alpha$-tilted, and universal sequence metrics \cite{miyamoto2026mismatched}. Mariona \textit{et al.} derive finite-length coding bounds via guesswork \cite{mariona2024finite}, while Tan and Joudeh derive refined asymptotic results for constant-composition guessing decoders with abandonment \cite{tanjoudeh2025abandon}. For discrete additive channels, the guessing-based formulations order noise sequences, whereas the metric-based universal-decoding and general-output-alphabet channel-coding formulations do not use observation-dependent correction permutations. The decoder instead observes $R^n$, which may be continuous, and selects, for each observation, a deterministic correction permutation over a family of conditional posteriors. The finite-block criterion is pointwise minimax posterior-normalized rank regret, while exponent preservation is established separately under conditional tilt-closure and Shtarkov-complexity conditions.

ORBGRAND constructs an observation-dependent query order from the ranks of reliability magnitudes \cite{duffy2022ordered}. For antipodal signaling over additive white Gaussian noise, Liu \textit{et al.} derive an achievable rate for ORBGRAND with i.i.d.\ random codebooks and compare it with the channel capacity \cite{liu2022orbgrand}. CDF-ORBGRAND applies inverse-CDF companding to reliability ranks and achieves symmetric capacity for binary-input memoryless channels \cite{li2025rankcompanding}. Li and Zhang derive an ORBGRAND-specific random-coding union bound and a second-order achievable-rate expansion \cite{li2026orbfinite}. Wan and Zhang derive exact block-error, stopping-time, and average-test-count expressions and show that ordering error patterns by nonincreasing average guessing posterior is optimal over the error-pattern set under consideration \cite{wan2026orbtype}. These analyses use a specified channel model or induced reliability distribution; their ordering criteria differ from pointwise minimax optimization over a family of conditional posteriors.

GRAND-MO, ORBGRAND-AI, and SGRAND-ISI construct correction orders that exploit specified channel memory or correlation structure \cite{an2022keep,duffy2023orbgrandai,li2026isi}. GRAND-MO uses a Markov noise model, ORBGRAND-AI modifies reliability-based ordering to exploit channel correlation, and SGRAND-ISI constructs an ML-equivalent order for linear Gaussian intersymbol-interference channels. LP-GRAND gives exact likelihood-ordered enumeration for specified correlated Gaussian models, TC-GRAND computes posterior weights and tail masses under a specified finite-memory posterior, and the switching-and-drift formulation uses state-path mixtures for switching and pilot-based tracking for drift \cite{razeghi2026lpgrand,razeghi2026tail,razeghi2026switchdrift}. Wiame \textit{et al.} address fading-channel estimation errors by testing multiple channel-estimate candidates that provide soft inputs to ORBGRAND \cite{wiame2025channel}. These methods construct query orders from specified, estimated, or tracked channel or posterior information; their ordering criteria differ from pointwise minimax optimization over a prescribed family of conditional posteriors.

Sundaresan establishes bounds relating guessing functions to length functions satisfying Kraft's inequality \cite{sundaresan2007length}. When the Shtarkov normalizer is finite, normalized maximum likelihood minimizes worst-case log-loss regret over probability assignments \cite{shtarkov1987}. Liu \textit{et al.} extend this minimax-regret framework to sequential probability assignment with contexts through contextual Shtarkov sums and contextual normalized maximum likelihood \cite{liu2024contextual}. Jia \textit{et al.} derive minimax-regret bounds for sequential probability assignment with and without side information in terms of sequential square-root entropy \cite{jia2025minimax}. These probability-assignment formulations optimize probability assignments under log loss, whereas the finite-block problem optimizes a permutation of candidate corrections under posterior-normalized rank regret.

\section{Channel Model and Posterior-Ordered Soft GRAND}
\label{sec:model}

\subsection{Finite input, general soft observation}

Let $(\A,\oplusA)$ be a finite Abelian group with $|\A|=a\ge2$. All logarithms, entropies, and divergences are in base $a$ unless stated otherwise, and rates are measured in base-$a$ units per channel use. A length-$n$ codebook is a subset $\C_n\subseteq\A^n$ of size $M_n$. Messages are equiprobable, so, conditional on $\C_n$, the transmitted codeword $X^n$ is uniform on $\C_n$. Let $\Wcal$ denote the prescribed channel class. Fix a labeling $\A=\{a_1,\ldots,a_{|\A|}\}$. For every $n$, let $\prec_n$ be the induced lexicographic total order on $\A^n$. Whenever an ordering of corrections is induced by posterior, codelength, assignment, envelope, or score values, ties are resolved by $\prec_n$.

The receiver observes $R^n$ taking values in a standard Borel space $\Rcal_n$. We allow $R^n$ to contain all receiver-side information used by the decoder, including channel outputs, channel-state estimates, and pilot-derived statistics. Deterministic receiver preprocessing can be represented by a statistic $S^n=T_n(R^n)$. No conditional independence is assumed among the components of a composite observation $R^n$. When a one-letter memoryless model is considered later, $R$ denotes the one-letter receiver observation; $R_{\rm c}$ denotes the normalized code rate.

For each $W\in\Wcal$ and blocklength $n$, let $W_n(\mathrm d r^n\mid x^n)$ denote the corresponding stochastic kernel from $\A^n$ to $\Rcal_n$. For every $x^n\in\A^n$, assume that $W_n(\mathrm d r^n\mid x^n) = w_{W,n}(r^n\mid x^n)\,\mu_n(\mathrm d r^n)$, where $w_{W,n}(\cdot\mid x^n)$ is a selected measurable density and $\mu_n$ is a common $\sigma$-finite measure on $\Rcal_n$. These density representatives are fixed throughout and form part of the uncertainty model. Pointwise posterior envelopes are defined relative to these representatives. For an uncountable channel class, channel-dependent $\mu_n$-null sets need not be contained in one common $\mu_n$-null set; hence a pointwise envelope is not determined solely by the kernels' almost-everywhere equivalence classes. The global finite-block minimax criterion below uses the $\mu_n$-essential supremum and is unchanged if all selected representatives are modified on a common $\mu_n$-null set.

Whenever a pointwise posterior supremum over an uncountable class is used, we assume that the resulting envelope is measurable. This is automatic for finite or countable classes. If the parameter space is a separable metric space and, on a common observation domain, the posterior is measurable in the observation for each parameter and continuous in the parameter for each observation, then the envelope is measurable because its supremum equals the supremum over a countable dense parameter subset. The generalized-Gaussian family in Section~\ref{sec:signalexamples} satisfies these conditions.

A deterministic, codebook-independent reference decision is a measurable map $h_n\colon\Rcal_n\to\A^n$. It need not be matched to the unknown channel. Given $r^n$, every candidate word $x^n\in\A^n$ has the unique correction representation
\begin{equation}
e^n=x^n\ominusA h_n(r^n),
\qquad
x^n=h_n(r^n)\oplusA e^n.
\label{eq:bijection}
\end{equation}
For the transmitted codeword, define $E^n=X^n\ominusA h_n(R^n)$. For binary BPSK with $h_n$ given by a symbolwise hard-decision slicer, $E_i=1$ if and only if the slicer output differs from $X_i$.

\subsection{Uniform-input correction posterior}

To define a codebook-independent correction posterior, introduce an auxiliary uniform input $U^n\sim\mathrm{Unif}(\A^n)$, with $R^n$ generated according to $W_n(\cdot\mid U^n)$, and define $E_U^n=U^n\ominusA h_n(R^n)$. For a fixed channel $W$ and an observation $r^n$ for which the denominator below is positive, define
\begin{equation}
p_{W,n}(e^n\mid r^n) = \frac{w_{W,n}(r^n\mid h_n(r^n)\oplusA e^n)}{\sum_{u^n\in\A^n}w_{W,n}(r^n\mid h_n(r^n)\oplusA u^n)}.
\label{eq:posterior}
\end{equation}
This ratio is a version of the conditional pmf $\Prb\{E_U^n=e^n\mid R^n=r^n\}$ under the auxiliary uniform-input distribution. Since the denominator in \eqref{eq:posterior} is independent of $e^n$, the posterior probabilities and the likelihoods of the corresponding candidate words induce the same weak ordering.

\begin{proposition}[Matched posterior ordering yields ML decoding]
\label{prop:matchedml}
Fix $W$, an observation $r^n$ for which the denominator in \eqref{eq:posterior} is positive, a codebook $\C_n$, and the reference map $h_n$. Order $e^n\in\A^n$ in nonincreasing
$p_{W,n}(e^n\mid r^n)$ and query the candidate word $h_n(r^n)\oplusA e^n$ for membership in $\C_n$. The first codeword found is an ML codeword,
\begin{equation}
\hat x_{\rm ML}^n\in \argmax_{c^n\in\C_n}w_{W,n}(r^n\mid c^n).
\end{equation}
\end{proposition}

\begin{proof}
By \eqref{eq:bijection}, every $c^n\in\C_n$ corresponds to the unique correction $e^n=c^n\ominusA h_n(r^n)$. For fixed $r^n$, $p_{W,n}(e^n\mid r^n)$ equals $w_{W,n}(r^n\mid c^n)$ multiplied by the same positive normalizing factor for every $c^n\in\C_n$. Hence the posterior and likelihood weak orderings coincide, and the first queried correction producing a codeword maximizes the likelihood over $\C_n$.
\end{proof}

\begin{remark}[Posterior normalization]
For fixed $r^n$, multiplying all candidate likelihoods by the same positive factor does not change their ordering. The normalization in \eqref{eq:posterior} produces a probability mass function on the finite correction space $\A^n$ without changing this ordering. This normalization is used below in the definitions of Shtarkov normalization and log-loss regret.
\end{remark}

\begin{remark}[Metric realizations and ties]
More generally, any measurable codebook-independent score $s_n(x^n,r^n)$ induces a correction order by ordering $s_n(h_n(r^n)\oplusA e^n,r^n)$ in nonincreasing value, with a fixed deterministic tie-breaking rule. The first queried candidate belonging to $\C_n$ is then an element of $\argmax_{c^n\in\C_n}s_n(c^n,r^n)$. A metric decoder that maximizes the same score over $\C_n$ and uses the same candidate-word tie-breaking order agrees with this query decoder. Different tie-breaking rules may select different elements of the same argmax set.
\end{remark}

\begin{proposition}[Reference-decision invariance]
\label{prop:href}
Fix $W$ and let $h_n$ and $\widetilde h_n$ be two measurable, codebook-independent reference maps. For every observation $r^n$ for which \eqref{eq:posterior} is defined, the correction posteriors induced by $h_n$ and $\widetilde h_n$ differ only by a permutation of $\A^n$. Hence they have the same multiset of posterior probabilities. Under the auxiliary uniform-input distribution, define $E_{U,h}^n=U^n\ominusA h_n(R^n)$ and $E_{U,\widetilde h}^n=U^n\ominusA\widetilde h_n(R^n)$. Then
\begin{equation}
H(E_{U,h}^n\mid R^n) = H(E_{U,\widetilde h}^n\mid R^n) = H(U^n\mid R^n).
\label{eq:entinv}
\end{equation}
The matched correction orderings induced by the two reference maps therefore differ only by the corresponding relabeling, up to permutations within posterior-tie classes. This statement does not concern the codebook-restricted posterior for an arbitrary fixed deterministic codebook.
\end{proposition}

\begin{proof}
For each fixed $r^n$, both maps $u^n\mapsto u^n\ominusA h_n(r^n)$ and $u^n\mapsto u^n\ominusA\widetilde h_n(r^n)$ are bijections on $\A^n$. Their composition therefore permutes the correction labels and the corresponding posterior probabilities. Moreover, each auxiliary correction is a bijective function of $U^n$ for fixed $R^n$, with inverse determined by $R^n$. Conditional entropy is therefore preserved, which gives \eqref{eq:entinv}.
\end{proof}

Therefore, the reference map need not depend on the unknown channel. Changing the reference map relabels the corrections but leaves the likelihood weak ordering of the corresponding candidate words unchanged.

\subsection{Universal Soft GRAND}

A deterministic, measurable, codebook-independent soft query rule is a family of bijections $\pi_n(r^n)\colon $ $\{1,\ldots,a^n\}\to\A^n$ such that $r^n\mapsto\pi_n(r^n)(j)$ is measurable for every $j\in\{1,\ldots,a^n\}$. Its rank function is
\begin{equation}
G_{\pi_n}(e^n\mid r^n)=\pi_n(r^n)^{-1}(e^n).
\label{eq:rankdef}
\end{equation}
Universal Soft GRAND queries
\begin{equation}
c_j^n(r^n)=h_n(r^n)\oplusA\pi_n(r^n)(j),
\qquad j=1,\ldots,a^n,
\end{equation}
and stops at the smallest $j$ for which $c_j^n(r^n)\in\C_n$. The query order may depend on the prescribed channel class $\Wcal$ and on the observation $r^n$, but is independent of the unknown channel $W\in\Wcal$ and of the codebook $\C_n$. The codebook enters only through the membership test.
For the transmitted codeword, define the correction rank
\begin{equation}
D_n=G_{\pi_n}(E^n\mid R^n).
\label{eq:Dn}
\end{equation}
For the auxiliary uniform-input experiment, define
\begin{equation}
D_{U,n}=G_{\pi_n}(E_U^n\mid R^n).
\label{eq:DUn}
\end{equation}

\begin{remark}[Distribution used for rank statements]
\label{rem:uniformbenchmark}
For a fixed codebook and observation, $p_{W,n}(\cdot\mid r^n)$ is a codebook-independent normalization of the candidate likelihoods. Its nonincreasing ordering therefore agrees with the likelihood weak ordering, and Proposition~\ref{prop:matchedml} yields an ML codeword. Distributional statements involving the correction posterior, posterior self-information, or query rank are evaluated under the auxiliary experiment with $U^n\sim\mathrm{Unif}(\A^n)$ and $R^n$ generated according to $W_n(\cdot\mid U^n)$, except where another distribution is specified. Under the uniform-subset ensemble used later, averaging over the random codebook and the uniform transmitted message makes $X^n$ uniform on $\A^n$. Consequently, the ensemble-averaged joint distribution of $(X^n,R^n)$ equals the auxiliary joint distribution of $(U^n,R^n)$. Because $h_n$ and $\pi_n$ are codebook-independent, the ensemble-averaged joint distribution of $(E^n,R^n,D_n)$ equals the joint distribution of $(E_U^n,R^n,D_{U,n})$ under the auxiliary experiment. For a fixed deterministic codebook, however, the conditional distribution of $X^n$ given $R^n$ is supported on $\C_n$, and the induced conditional distribution of $E^n$ need not equal $p_{W,n}(\cdot\mid R^n)$.
\end{remark}

\section{Conditional Query Rank and First-Order Universality}
\label{sec:firstorder}

\subsection{Conditional Posterior AEP}

For a fixed channel $W$, let $(U^n,R^n)$ have the joint distribution induced by $U^n\sim\mathrm{Unif}(\A^n)$ and $W_n$, and define $E_U^n=U^n\ominusA h_n(R^n)$. The set of observations at which the normalizer in \eqref{eq:posterior} vanishes has zero probability under the auxiliary output distribution. For the distributional statements in this section, extend $p_{W,n}(\cdot|r^n)$ measurably to this set, for example by assigning the uniform pmf on $\A^n$. This convention does not change the active-channel sets or pointwise posterior envelopes used in the finite-block minimax criterion. We write $\Prb_W$, $\E_W$, and $H_W$ for probability, expectation, and entropy under the auxiliary uniform-input distribution associated with $W$. Define the posterior self-information
\begin{equation}
\jmath_{W,n}(e^n,r^n) \coloneqq -\log p_{W,n}(e^n|r^n),
\label{eq:jpost}
\end{equation}
with the convention $-\log 0=+\infty$, and let $\jmath_{W,n}\coloneqq\jmath_{W,n}(E_U^n,R^n)$.

\begin{assumption}[Conditional posterior AEP]
\label{assump:aep}
There exists a constant $h(W)$ such that
\begin{equation}
\frac{1}{n}\jmath_{W,n}\xrightarrow{P}h(W).
\label{eq:aep}
\end{equation}
\end{assumption}

Assumption~\ref{assump:aep} does not require memorylessness. If $(E_U^n,R^n)$ are the length-$n$ marginals of a finite-alphabet stationary ergodic pair process, the assumption follows by applying the Shannon--McMillan--Breiman theorem to the joint and marginal processes \cite{coverthomas2006}. In the symbolwise memoryless setting considered below, it follows because the sample average of the i.i.d. posterior self-information terms converges in probability to their common expectation. Any $h(W)$ satisfying \eqref{eq:aep} necessarily lies in $[0,1]$. Since $\jmath_{W,n}\ge0$, any probability limit in \eqref{eq:aep} is nonnegative. Moreover, for every $\delta>0$,
\begin{equation}
\Prb_W\!\left\{\jmath_{W,n}>n(1+\delta)\right\} \le a^{-n\delta}.
\end{equation}
To obtain this bound, condition on $R^n=r^n$ and write $p(e^n)=p_{W,n}(e^n|r^n)$. Then
\begin{equation*}
\sum_{e^n:\,p(e^n)<a^{-n(1+\delta)}}p(e^n)
\le a^n a^{-n(1+\delta)} = a^{-n\delta}.
\end{equation*}
If $h(W)>1$, choosing $0<\delta<h(W)-1$ would make the probability on the left converge to one by \eqref{eq:aep}, contradicting the displayed bound. Hence $h(W)\le1$.

\subsection{Conditional Kraft assignments}

\begin{definition}[Conditional Kraft-admissible codelength]
A measurable function $L_n\colon\A^n\times\Rcal_n\to[0,+\infty]$ is conditionally Kraft-admissible if, for every $r^n\in\Rcal_n$,
\begin{equation}
\sum_{e^n\in\A^n}a^{-L_n(e^n|r^n)}\le1.
\label{eq:condkraft}
\end{equation}
It induces an order $\pi_n^L(r^n)$ by nondecreasing $L_n(e^n|r^n)$, with the prescribed deterministic tie-breaking rule.
\end{definition}

If $Q_n$ is a measurable conditional pmf on $\A^n$ given $\Rcal_n$, then $L_n(e^n|r^n)=-\log Q_n(e^n|r^n)$ satisfies \eqref{eq:condkraft} with equality, with the convention $-\log 0=+\infty$.

\begin{lemma}[Conditional Kraft rank bound]
\label{lem:kraftrank}
For every conditionally Kraft-admissible $L_n$,
\begin{equation}
G_{\pi_n^L}(e^n|r^n)\le a^{L_n(e^n|r^n)}
\label{eq:kraftrank}
\end{equation}
for all $(e^n,r^n)\in\A^n\times\Rcal_n$.
\end{lemma}

\begin{proof}
Fix $(e^n,r^n)\in\A^n\times\Rcal_n$. For this fixed $r^n$, the following is the length-function argument for guessing~\cite{sundaresan2007length} applied to $u^n\mapsto L_n(u^n|r^n)$. If $L_n(e^n|r^n)=+\infty$, then \eqref{eq:kraftrank} is immediate. Otherwise, let $B=\left\{u^n\in\A^n\,\middle|\,L_n(u^n|r^n)\le L_n(e^n|r^n)\right\}$. By \eqref{eq:condkraft}, $1 \ge \sum_{u^n\in B}a^{-L_n(u^n|r^n)} \ge |B|a^{-L_n(e^n|r^n)}$, and hence $|B|\le a^{L_n(e^n|r^n)}$. Since $\pi_n^L(r^n)$ orders corrections by nondecreasing codelength, for every $u^n\in\A^n$,
\[
G_{\pi_n^L}(u^n|r^n) \le G_{\pi_n^L}(e^n|r^n)
\quad\Longrightarrow\quad
L_n(u^n|r^n)\le L_n(e^n|r^n),
\]
and therefore $u^n\in B$. Thus $G_{\pi_n^L}(e^n|r^n) \le |B| \le a^{L_n(e^n|r^n)}$, which proves \eqref{eq:kraftrank}.
\end{proof}

\begin{definition}[Posterior-universal conditional assignment]
\label{def:posterioruniv}
For a fixed $W$, a sequence of measurable conditional pmfs $\{Q_n\}$ is posterior-universal at $W$ if
\begin{equation}
\frac{1}{n} \log\frac{p_{W,n}(E_U^n|R^n)}{Q_n(E_U^n|R^n)}
\xrightarrow{P}0
\label{eq:postuniv}
\end{equation}
under the auxiliary uniform-input distribution associated with $W$. It is posterior-universal over $\Wcal$ if \eqref{eq:postuniv} holds for every $W\in\Wcal$. The logarithmic ratio is $+\infty$ whenever $Q_n(E_U^n|R^n)=0<p_{W,n}(E_U^n|R^n)$.
\end{definition}

\begin{lemma}[Negative-tail bound for posterior regret]
\label{lem:regretnegativetail}
Let $R^n$ have an arbitrary distribution, let $p_n(\cdot|r^n)$ and $Q_n(\cdot|r^n)$ be conditional pmfs on $\A^n$, and let $Z^n|R^n=r^n\sim p_n(\cdot|r^n)$. For every $t>0$,
\begin{equation}
\Prb\!\left\{ \log\frac{p_n(Z^n|R^n)}{Q_n(Z^n|R^n)}<-t \right\}
\le a^{-t}.
\label{eq:regretnegativetail}
\end{equation}
The ratio is interpreted as $+\infty$ when $Q_n(Z^n|R^n)=0<p_n(Z^n|R^n)$; no absolute-continuity assumption of $p_n(\cdot|r^n)$ with respect to $Q_n(\cdot|r^n)$ is required.
\end{lemma}

\begin{proof}
Condition on $R^n=r^n$ and define
\[
A_t(r^n) = \left\{ e^n\in\A^n \,\middle|\, p_n(e^n|r^n)<a^{-t}Q_n(e^n|r^n) \right\}.
\]
Then
\begin{equation}
\sum_{e^n\in A_t(r^n)}p_n(e^n|r^n) \le a^{-t} \sum_{e^n\in A_t(r^n)}Q_n(e^n|r^n)
\le a^{-t}.
\end{equation}
Averaging over $R^n$ gives \eqref{eq:regretnegativetail}.
\end{proof}

Consequently, to establish \eqref{eq:postuniv}, it is sufficient to control the positive tail of the normalized posterior regret, because, for every $\epsilon>0$,
\begin{equation}
\Prb_W\!\left\{ \frac{1}{n} \log\frac{p_{W,n}(E_U^n|R^n)}{Q_n(E_U^n|R^n)} <-\epsilon \right\}
\le a^{-n\epsilon}.
\end{equation}
Under Assumption~\ref{assump:aep}, posterior universality at $W$ therefore implies
\begin{equation}
-\frac{1}{n}\log Q_n(E_U^n|R^n) \xrightarrow{P}h(W).
\label{eq:Qlengthconv}
\end{equation}

\subsection{Converse for observation-dependent orderings}

For a sequence of real-valued random variables $\{Z_n\}$, define
\[
\pliminf_{n\to\infty} Z_n \coloneqq \sup\left\{z\in\mathbb R: \Prb\{Z_n<z\}\to0\right\},
\qquad
\plimsup_{n\to\infty} Z_n \coloneqq \inf\left\{z\in\mathbb R: \Prb\{Z_n>z\}\to0\right\}.
\]

\begin{lemma}[Lower bound on query-rank exponent]
\label{lem:converse}
Suppose Assumption~\ref{assump:aep} holds for $W$. For every deterministic, measurable, codebook-independent observation-dependent ordering sequence $\boldsymbol\pi=\{\pi_n\}$,
\begin{equation}
\pliminf_{n\to\infty}\frac{1}{n}\log G_{\pi_n}(E_U^n|R^n) \ge h(W).
\label{eq:rankconverse}
\end{equation}
\end{lemma}

\begin{proof}
If $h(W)=0$, the result follows from $G_{\pi_n}(E_U^n|R^n)\ge1$. Suppose $h(W)>0$. Fix $0<\delta<h(W)$ and choose $0<\eta<\delta$. Let
\[
 t_{n,\delta}=\ceil{a^{n(h(W)-\delta)}},
\]
and let $A_n(r^n)$ denote the first $t_{n,\delta}$ corrections under $\pi_n(r^n)$. Define
\[
B_n(r^n) = \left\{ e^n\in\A^n \,\middle|\, p_{W,n}(e^n|r^n)>a^{-n(h(W)-\eta)} \right\}.
\]
By Assumption~\ref{assump:aep},
\begin{equation}
\Prb_W\{E_U^n\in B_n(R^n)\} = \Prb_W\{\jmath_{W,n}<n(h(W)-\eta)\} \to0.
\end{equation}
For every $r^n$,
\begin{align}
\sum_{e^n\in A_n(r^n)}p_{W,n}(e^n|r^n) \le
\sum_{e^n\in B_n(r^n)}p_{W,n}(e^n|r^n) + t_{n,\delta} a^{-n(h(W)-\eta)}.
\end{align}
Since $t_{n,\delta}=a^{n(h(W)-\delta)+o(n)}$, averaging over $R^n$ gives
\begin{align}
\Prb_W\{D_{U,n}\le t_{n,\delta}\} \le \Prb_W\{E_U^n\in B_n(R^n)\} + a^{-n(\delta-\eta)+o(n)} \to0.
\end{align}
Hence
\[
\pliminf_{n\to\infty} \frac{1}{n}\log D_{U,n}  \ge h(W)-\delta.
\]
Letting $\delta\downarrow0$ proves \eqref{eq:rankconverse}.
\end{proof}

\begin{theorem}[First-order query-rank exponent]
\label{thm:exactrank}
Suppose Assumption~\ref{assump:aep} holds for $W$ and $\{Q_n\}$ is posterior-universal at $W$. Let $\pi_n^Q$ order corrections in nonincreasing $Q_n(\cdot|r^n)$, using the prescribed deterministic tie-breaking rule. Then
\begin{equation}
\frac{1}{n}\log G_{\pi_n^Q}(E_U^n|R^n) \xrightarrow{P}h(W).
\label{eq:exactrank}
\end{equation}
\end{theorem}

\begin{proof}
With $L_n=-\log Q_n$, Lemma~\ref{lem:kraftrank} and \eqref{eq:Qlengthconv} give
\[
\plimsup_{n\to\infty} \frac{1}{n}\log G_{\pi_n^Q}(E_U^n|R^n) \le h(W).
\]
Lemma~\ref{lem:converse} gives
\[
\pliminf_{n\to\infty} \frac{1}{n}\log G_{\pi_n^Q}(E_U^n|R^n) \ge h(W).
\]
\end{proof}

\begin{remark}[Information-spectrum form]
\label{rem:infospectrum}
The AEP is not required for the underlying spectral rank bounds. Define
\begin{align}
\underline h(W) &= \pliminf_{n\to\infty} -\frac1n\log p_{W,n}(E_U^n|R^n),\\
\overline h(W) &= \plimsup_{n\to\infty} -\frac1n\log p_{W,n}(E_U^n|R^n).
\end{align}
For every deterministic, measurable, codebook-independent ordering sequence $\boldsymbol\pi=\{\pi_n\}$, the proof of Lemma~\ref{lem:converse} gives
\begin{equation}
\pliminf_{n\to\infty}\frac1n \log G_{\pi_n}(E_U^n|R^n) \ge \underline h(W).
\end{equation}
If a sequence of measurable conditional pmfs $\{Q_n\}$ satisfies
\begin{equation}
\plimsup_{n\to\infty}\frac1n \log\frac{p_{W,n}(E_U^n|R^n)}{Q_n(E_U^n|R^n)}
\le0,
\end{equation}
then Lemma~\ref{lem:kraftrank} gives
\begin{equation}
\plimsup_{n\to\infty}\frac1n \log G_{\pi_n^Q}(E_U^n|R^n) \le \overline h(W).
\end{equation}
\end{remark}

\begin{definition}[Minimax soft query-rank exponent]
For a class $\Wcal$ for which Assumption~\ref{assump:aep} holds for every $W\in\Wcal$, define
\begin{equation}
q^\star(\Wcal) = \inf_{\boldsymbol\pi} \sup_{W\in\Wcal} \plimsup_{n\to\infty} \frac{1}{n}\log
G_{\pi_n}(E_U^n|R^n),
\label{eq:qstar}
\end{equation}
where the infimum is over deterministic, measurable, codebook-independent ordering sequences $\boldsymbol\pi=\{\pi_n\}$ that may depend on $\Wcal$ but not on the unknown $W\in\Wcal$.
\end{definition}

\begin{theorem}[Minimax first-order soft query-rank exponent]
\label{thm:minimax}
If there exists a sequence $\{Q_n\}$ that is posterior-universal over $\Wcal$, then
\begin{equation}
q^\star(\Wcal)=\sup_{W\in\Wcal}h(W),
\label{eq:minimax}
\end{equation}
and the ordering sequence induced by $\{Q_n\}$ attains the infimum in \eqref{eq:qstar}.
\end{theorem}

\begin{proof}
For every admissible ordering sequence and every $W\in\Wcal$, Lemma~\ref{lem:converse} gives
\[
\plimsup_{n\to\infty} \frac{1}{n}\log G_{\pi_n}(E_U^n|R^n) \ge h(W).
\]
Hence $q^\star(\Wcal)\ge\sup_{W\in\Wcal}h(W)$. For the ordering induced by
the common sequence $\{Q_n\}$, Theorem~\ref{thm:exactrank} gives equality
for every $W\in\Wcal$, proving the reverse inequality.
\end{proof}

In particular, if $\Wcal$ contains a channel $W$ for which $R^n$ is independent of $U^n$ for every $n$, then $p_{W,n}(E_U^n|R^n)=a^{-n}$ almost surely and hence $h(W)=1$. Lemma~\ref{lem:converse} and the bound $G_{\pi_n}\le a^n$ then give $q^\star(\Wcal)=1$. More generally, the converse implies that $q^\star(\Wcal)<1$ requires $\sup_{W\in\Wcal}h(W)<1$.

\subsection{Memoryless channels and soft-information compression}

Consider a stationary memoryless channel $W$ on a standard Borel one-letter observation space $\Rcal$ such that $W(\mathrm d r|x)=w_W(r|x)\,\mu(\mathrm d r)$, $x\in\A$, for a $\sigma$-finite measure $\mu$ on $\Rcal$. Take $\Rcal_n=\Rcal^n$, $\mu_n=\mu^{\otimes n}$, and $w_{W,n}(r^n|x^n) = \prod_{i=1}^n w_W(r_i|x_i)$. 
Let the reference decision be symbolwise, $h_n(r^n)=(h(r_1),\ldots,h(r_n))$. Let $(X,R)$ denote the one-letter uniform-input channel pair, with $X\sim\mathrm{Unif}(\A)$, and define $E=X\ominusA h(R)$. Then the auxiliary pairs $(E_{U,i},R_i)$ are i.i.d. copies of $(E,R)$. If $p_W(e|r)=\Prb_W\{E=e|R=r\}$ is any fixed measurable version of the one-letter correction posterior, then, for $P_{W,R^n}$-almost every $r^n$,
\begin{equation}
p_{W,n}(e^n|r^n) = \prod_{i=1}^n p_W(e_i|r_i),
\qquad e^n\in\A^n.
\end{equation}
We write $H_{\mathrm U,W}$ and $I_{\mathrm U,W}$ for entropy and mutual information under this one-letter uniform-input distribution. The sample average of the i.i.d. posterior self-information terms converges in probability to $H_{\mathrm U,W}(E|R)$. Therefore,
\begin{equation}
h(W) = H_{\mathrm U,W}(E|R) = H_{\mathrm U,W}(X|R).
\label{eq:hwmemoryless}
\end{equation}

\begin{corollary}[Soft-information query exponent]
\label{cor:softentropy}
For a memoryless channel $W$ under uniform input, any assignment sequence that is posterior-universal at $W$ induces an order satisfying
\begin{equation}
\frac{1}{n}\log G_{\pi_n^Q}(E_U^n|R^n) \xrightarrow{P} H_{\mathrm U,W}(X|R).
\end{equation}
Moreover, every deterministic, measurable, codebook-independent observation-dependent ordering sequence $\boldsymbol\pi=\{\pi_n\}$ satisfies
\begin{equation}
\pliminf_{n\to\infty} \frac{1}{n}\log G_{\pi_n}(E_U^n|R^n)
\ge H_{\mathrm U,W}(X|R).
\end{equation}
\end{corollary}

\begin{corollary}[Cost of soft-information compression]
\label{cor:quantization}
Consider a memoryless channel $W$ under uniform input. Let $S=T(R)$ be a deterministic receiver statistic, applied symbolwise so that
$S^n=(T(R_1),\ldots,T(R_n))$, and suppose that $h(R)$ is $S$-measurable. Among all deterministic, codebook-independent ordering sequences measurable with respect to $S^n$, the minimum first-order query-rank exponent is $H_{\mathrm U,W}(X|S)$. For fixed $W$, this minimum is attained by the matched order induced by $P_W(E^n|S^n)$. Hence replacing $R$ by $S$ increases the minimum exponent by
\begin{align}
H_{\mathrm U,W}(X|S)-H_{\mathrm U,W}(X|R)
&= I_{\mathrm U,W}(X;R|S)\\
&= I_{\mathrm U,W}(X;R)-I_{\mathrm U,W}(X;S)
\ge0.
\label{eq:quantloss}
\end{align}
For a class $\Wcal$ of memoryless channels, a sufficient condition for a single $S$-based ordering sequence to attain $H_{\mathrm U,W}(X|S)$ for every $W\in\Wcal$ is that the family of conditional correction posteriors induced by $W\in\Wcal$ given $S^n$ admits a common posterior-universal assignment sequence, with $S^n$ in place of $R^n$ in Definition~\ref{def:posterioruniv}.
\end{corollary}

\subsection{One-shot posterior-rank bounds and sufficient statistics}
\label{subsec:oneshot}

Finite-block bounds relate query rank to posterior self-information. Recall $\jmath_{W,n}$ from \eqref{eq:jpost}, and let $\mathsf H_N=\sum_{j=1}^N j^{-1}$.

\begin{theorem}[One-shot posterior-rank bounds]
\label{thm:oneshotrank}
Fix $W$ and $n$. 
\begin{enumerate}[label=(\roman*),leftmargin=*]
\item 
For every deterministic, measurable, codebook-independent observation-dependent order $\pi_n$ and every $t>0$,
\begin{align}
\Prb_W\!\left\{ \log G_{\pi_n}(E_U^n|R^n) \le \jmath_{W,n}-t \right\}
\le \mathsf H_{a^n}a^{-t}.
\label{eq:oneshotlower}
\end{align}
\item 
Let $Q_n$ be a measurable conditional pmf and define
\begin{align}
r_n(W,Q_n) = \esssup_{r^n\in\Rcal_n}\; \max_{\substack{e^n\in\A^n\\ p_{W,n}(e^n|r^n)>0}}
\log \frac{p_{W,n}(e^n|r^n)}{Q_n(e^n|r^n)},
\label{eq:rnWQ}
\end{align}
where the essential supremum is with respect to the auxiliary uniform-input distribution of $R^n$, and the ratio is $+\infty$ when $Q_n(e^n|r^n)=0<p_{W,n}(e^n|r^n)$. For the order induced by
nonincreasing $Q_n(\cdot|r^n)$,
\begin{equation}
\log G_{\pi_n^Q}(E_U^n|R^n) \le \jmath_{W,n}+r_n(W,Q_n)
\label{eq:oneshotupper}
\end{equation}
almost surely.
\item 
For every $\delta\in(0,1)$,
\begin{equation}
\jmath_{W,n} - \log\frac{\mathsf H_{a^n}}{\delta}
\le \log G_{\pi_n^Q}(E_U^n|R^n)
\le \jmath_{W,n}+r_n(W,Q_n)
\label{eq:oneshotsandwich}
\end{equation}
with probability at least $1-\delta$.
\end{enumerate}
\end{theorem}

\begin{proof}
For (i), fix $r^n\in\Rcal_n$ and write $p(e^n)=p_{W,n}(e^n|r^n)$ and $G(e^n)=G_{\pi_n}(e^n|r^n)$. The argument is the conditional version of the usual finite-alphabet rank-counting bound for guessing ~\cite{sundaresan2007length}. The event in \eqref{eq:oneshotlower} is equivalent to $G(e^n)p(e^n)\le a^{-t}$. Since $G$ is a bijection from $\A^n$ to $\{1,\ldots,a^n\}$,
\begin{align}
\sum_{e^n:\,G(e^n)p(e^n)\le a^{-t}}p(e^n) &\le
a^{-t} \sum_{e^n:\,G(e^n)p(e^n)\le a^{-t}} \frac{1}{G(e^n)} \\
& \le a^{-t}\mathsf H_{a^n}.
\end{align}
Integrating the conditional bound with respect to the distribution of $R^n$ proves (i).

For (ii), Lemma~\ref{lem:kraftrank} applied to $L_n=-\log Q_n$ gives $G_{\pi_n^Q}(e^n|r^n)Q_n(e^n|r^n)\le1$. Hence, whenever $p_{W,n}(e^n|r^n)>~\!0$,
\[
\log G_{\pi_n^Q}(e^n|r^n) 
\le -\log p_{W,n}(e^n|r^n) + \log \frac{p_{W,n}(e^n|r^n)}{Q_n(e^n|r^n)}.
\]
Equation~\eqref{eq:oneshotupper} follows from \eqref{eq:rnWQ}. Part (iii) follows from (i) with $t=\log(\mathsf H_{a^n}/\delta)$ and from (ii).
\end{proof}

For any deterministic, measurable, codebook-independent observation-dependent order $\pi_n$, define
\[
\ell_n(e^n|r^n) = \log G_{\pi_n}(e^n|r^n)+\log\mathsf H_{a^n}.
\]
For every $r^n$, $\sum_{e^n\in\A^n}a^{-\ell_n(e^n|r^n)}=1$. Thus $a^{-\ell_n(\cdot|r^n)}$ is a conditional pmf. The conditional cross-entropy inequality gives
\begin{equation}
\E_W\!\left[ \log G_{\pi_n}(E_U^n|R^n) \right]
\ge H_W(E_U^n|R^n)-\log\mathsf H_{a^n}.
\label{eq:meanranklower}
\end{equation}
For the order $\pi_n^Q$, Lemma~\ref{lem:kraftrank} with $L_n=-\log Q_n$ gives
\begin{equation}
\E_W\!\left[ \log G_{\pi_n^Q}(E_U^n|R^n) \right]
\le \E_W[-\log Q_n(E_U^n|R^n)].
\end{equation}
Hence
\begin{align}
H_W(E_U^n|R^n)-\log\mathsf H_{a^n} \le
\E_W\!\left[ \log G_{\pi_n^Q}(E_U^n|R^n) \right]
\le \E_W[-\log Q_n(E_U^n|R^n)].
\end{align}
Since $\log\mathsf H_{a^n}=\mathcal O(\log n)=o(n)$, the lower-bound correction is sublinear in $n$.


\begin{theorem}[Posterior-preserving receiver statistics]
\label{thm:sufficiency}
Consider a memoryless channel $W$ under the one-letter uniform-input distribution defined above. Let $S=T(R)$ be a deterministic receiver statistic and suppose the symbolwise reference decision $h(R)$ is $S$-measurable. The minimum first-order query-rank exponent based on $S$ equals that based on $R$ if and only if
\begin{equation}
\Prb_W\{X=x\mid R\}=\Prb_W\{X=x\mid S\} \quad \mathrm{a.s.}\quad \forall x\in\A,
\label{eq:suffposterior}
\end{equation}
where the conditional probabilities are evaluated under the uniform-input distribution. Equivalently, $I_{\mathrm U,W}(X;R|S)=0$. If \eqref{eq:suffposterior} holds for every $W\in\Wcal$, then, for every $W\in\Wcal$,
\begin{align}
H_{\mathrm U,W}(X|S)&=H_{\mathrm U,W}(X|R),\\
I_{\mathrm U,W}(X;S)&=I_{\mathrm U,W}(X;R).
\end{align}
\end{theorem}

\begin{proof}
Corollary~\ref{cor:quantization} gives $H_{\mathrm U,W}(X|S)-H_{\mathrm U,W}(X|R) = I_{\mathrm U,W}(X;R|S)$. Since $S$ is a deterministic function of $R$, the right-hand side is zero if
and only if $X-S-R$ is a Markov chain, which is equivalent to \eqref{eq:suffposterior}. The final identities follow for each $W\in\Wcal$.
\end{proof}

More generally, if $S$ takes at most $2^b$ values, then $I_{\mathrm U,W}(X;S) \le H_{\mathrm U,W}(S) \le \frac{b}{\log_2 a}$. Since $X$ is uniform on $\A$,
\begin{equation}
H_{\mathrm U,W}(X|S) \ge \left(1-\frac{b}{\log_2 a}\right)^+.
\label{eq:bbitbound}
\end{equation}

\section{Finite-Block Minimax Posterior-Normalized Rank Regret}
\label{sec:regret}

Multiplying a rank by a factor $a^{o(n)}$ changes its normalized logarithm by $o(1)$. We therefore consider the following finite-block criterion. Define
\begin{equation}
Z_{W,n}(r^n) = \sum_{u^n\in\A^n} w_{W,n}\!\left( r^n\middle| h_n(r^n)\oplusA u^n \right),
\end{equation}
and, for each observation,
\begin{equation}
\Wcal_n(r^n) = \left\{ W\in\Wcal \,\middle|\, Z_{W,n}(r^n)>0 \right\}.
\end{equation}
Let $\Rcal_n^+ = \left\{ r^n\in\Rcal_n \,\middle|\, \Wcal_n(r^n)\neq\varnothing \right\}$. We call $r^n\in\Rcal_n^+$ an admissible observation. For $r^n\in\Rcal_n^+$ and $W\in\Wcal_n(r^n)$, $p_{W,n}(\cdot|r^n)$ is defined by \eqref{eq:posterior}.
Assume that $\Rcal_n^+$ is measurable. For each $r^n\in\Rcal_n^+$, define the posterior envelope
\begin{equation}
m_n(e^n|r^n) = \sup_{W\in\Wcal_n(r^n)} p_{W,n}(e^n|r^n),
\label{eq:env}
\end{equation}
and assume that $r^n\mapsto m_n(e^n|r^n)$ is measurable on $\Rcal_n^+$ for every $e^n\in\A^n$. Let $m_{n,(1)}(r^n)\ge\cdots\ge m_{n,(a^n)}(r^n)$ denote the envelope values in nonincreasing order. Any conditional assignment or ordering rule defined on $\Rcal_n^+$ is extended by a fixed conditional pmf or fixed permutation, respectively, on $\Rcal_n\setminus\Rcal_n^+$.

\begin{definition}[Posterior-normalized rank regret]
For a deterministic, measurable, codebook-independent order $\pi_n$, define
\begin{align}
\mathfrak r_n(\pi_n,\Wcal) = \esssup_{r^n\in\Rcal_n^+} \sup_{W\in\Wcal_n(r^n)} \max_{\substack{e^n\in\A^n\\
p_{W,n}(e^n|r^n)>0}} \log\!\left(  G_{\pi_n}(e^n|r^n) p_{W,n}(e^n|r^n) \right).
\label{eq:rho}
\end{align}
Essential suprema over $r^n\in\Rcal_n^+$ are with respect to the restriction of $\mu_n$ to $\Rcal_n^+$. The minimax value is
\begin{equation}
\mathfrak r_n^\star(\Wcal) = \inf_{\pi_n}\mathfrak r_n(\pi_n,\Wcal),
\end{equation}
where the infimum is over deterministic, measurable, codebook-independent orders.
\end{definition}

For $p_{W,n}(e^n|r^n)>0$, the product $G_{\pi_n}(e^n|r^n)p_{W,n}(e^n|r^n)$ is the ratio of the query rank to the reciprocal posterior probability. The posterior-normalized rank regret may be negative. In particular, for some posterior distributions and query permutations, $G_{\pi_n}(e^n|r^n)p_{W,n}(e^n|r^n)<1$ for every correction with positive posterior probability.
For a fixed $r^n\in\Rcal_n^+$, suppose that
\[
\sup_{W\in\Wcal_n(r^n)} \max_{\substack{e^n\in\A^n\\p_{W,n}(e^n|r^n)>0}} \log\!\left( G_{\pi_n}(e^n|r^n)p_{W,n}(e^n|r^n) \right) \le b.
\]
Then, for every $W\in\Wcal_n(r^n)$ and every $e^n$ such that $p_{W,n}(e^n|r^n)>0$,
\[
G_{\pi_n}(e^n|r^n) \le \frac{a^b}{p_{W,n}(e^n|r^n)}.
\]
For GRAND without abandonment, the number of queries until the first codeword hit is no greater than the rank of the realized correction. Hence the criterion in \eqref{eq:rho} provides a uniform finite-block upper bound on the query count relative to the reciprocal posterior probability. This bound does not imply preservation of the matched error exponent.

\begin{theorem}[Exact minimax posterior-rank ordering]
\label{thm:exactregret}
For each $r^n\in\Rcal_n^+$, define
\begin{align}
\mathfrak r_n^\star(r^n) &= \min_{\pi_n(r^n)} \sup_{W\in\Wcal_n(r^n)}
\max_{\substack{e^n\in\A^n\\p_{W,n}(e^n|r^n)>0}}
\log \left( G_{\pi_n}(e^n|r^n)p_{W,n}(e^n|r^n) \right),
\end{align}
where the minimum is over permutations of $\A^n$. Every nonincreasing ordering of $m_n(\cdot|r^n)$ attains this minimum, and
\begin{equation}
\mathfrak r_n^\star(r^n) = \log\max_{1\le j\le a^n}j\,m_{n,(j)}(r^n).
\label{eq:rhoobs}
\end{equation}
With the prescribed deterministic tie-breaking rule, the resulting observation-dependent envelope order is measurable. Consequently,
\begin{equation}
\mathfrak r_n^\star(\Wcal) = \esssup_{r^n\in\Rcal_n^+} \mathfrak r_n^\star(r^n).
\label{eq:rhoglobal}
\end{equation}
\end{theorem}

\begin{proof}
Fix $r^n\in\Rcal_n^+$. Since the logarithm is strictly increasing, it suffices to minimize the corresponding product criterion. For any permutation $\pi_n(r^n)$,
\begin{align}
\sup_{W\in\Wcal_n(r^n)} \max_{\substack{e^n\in\A^n\\p_{W,n}(e^n|r^n)>0}}
G_{\pi_n}(e^n|r^n)p_{W,n}(e^n|r^n)
= \max_{e^n\in\A^n} G_{\pi_n}(e^n|r^n)m_n(e^n|r^n).
\end{align}
Let $m_{(j)}=m_{n,(j)}(r^n)$. Under a nonincreasing-envelope order, the envelope value at rank $j$ is $m_{(j)}$. Hence
\begin{equation}
\max_{e^n\in\A^n} G_{\pi_n}(e^n|r^n)m_n(e^n|r^n) = \max_{1\le j\le a^n}j\,m_{(j)}.
\end{equation}
Conversely, consider any permutation and fix $j$. Among the $j$ corrections with largest envelope values, at least one has rank at least $j$. Its envelope value is at least $m_{(j)}$. Therefore,
\begin{equation}
\max_{e^n\in\A^n} G_{\pi_n}(e^n|r^n)m_n(e^n|r^n)
\ge j\,m_{(j)}.
\end{equation}
Maximizing over $j$ proves \eqref{eq:rhoobs}. For every deterministic, measurable, codebook-independent order,
\begin{equation}
\mathfrak r_n(\pi_n,\Wcal) \ge \esssup_{r^n\in\Rcal_n^+} \mathfrak r_n^\star(r^n).
\end{equation}
The measurable nonincreasing-envelope order attains $\mathfrak r_n^\star(r^n)$ at every $r^n\in\Rcal_n^+$. Taking the infimum over such orders proves \eqref{eq:rhoglobal}.
\end{proof}

\begin{corollary}[Query-count bound]
\label{cor:minimaxquerycount}
Let $Q_n^{\rm dec}$ denote the number of GRAND queries until termination. Fix $r^n\in\Rcal_n^+$, and let $\pi_n(r^n)$ be any nonincreasing ordering of $m_n(\cdot|r^n)$. For every $W\in\Wcal_n(r^n)$ and every $e^n\in\A^n$ such that $p_{W,n}(e^n|r^n)>0$, on a decoding realization with observation $r^n$ and correction $e^n$,
\[
Q_n^{\rm dec} \le G_{\pi_n}(e^n|r^n) \le \frac{a^{\mathfrak r_n^\star(r^n)}}{p_{W,n}(e^n|r^n)}.
\]
Conversely, for every permutation $\pi_n(r^n)$ and every $b<\mathfrak r_n^\star(r^n)$, there exist $W\in\Wcal_n(r^n)$ and $e^n\in\A^n$ with $p_{W,n}(e^n|r^n)>0$ such that $G_{\pi_n}(e^n|r^n)p_{W,n}(e^n|r^n)>a^b$.
\end{corollary}

\subsection{Posterior-envelope NML and Shtarkov regret}

For each $r^n\in\Rcal_n^+$, define the conditional Shtarkov normalizer
\begin{equation}
C_n(r^n) = \sum_{e^n\in\A^n}m_n(e^n|r^n).
\label{eq:Csht}
\end{equation}
Since $\Wcal_n(r^n)\neq\varnothing$, $1\le C_n(r^n)\le a^n$. Define the posterior-envelope NML assignment
\begin{equation}
Q_n^{\rm peNML}(e^n|r^n) = \frac{m_n(e^n|r^n)}{C_n(r^n)}.
\label{eq:penml}
\end{equation}
By the assumed envelope measurability, $Q_n^{\rm peNML}$ is a measurable conditional pmf on $\Rcal_n^+$ and is extended as specified above outside $\Rcal_n^+$. For fixed $r^n\in\Rcal_n^+$, it is the Shtarkov NML assignment for the active posterior family $\{p_{W,n}(\cdot|r^n):W\in\Wcal_n(r^n)\}$ \cite{shtarkov1987}.

\begin{theorem}[Posterior-rank regret and Shtarkov regret]
\label{thm:shtbridge}
Recall that $\mathsf H_N=\sum_{j=1}^N j^{-1}$. For every $r^n\in\Rcal_n^+$,
\begin{equation}
0 \le \log C_n(r^n)- \mathfrak r_n^\star(r^n) \le \log\mathsf H_{a^n}.
\label{eq:shtbridgeobs}
\end{equation}
Define
\begin{equation}
\bar r_n(\Wcal) = \esssup_{r^n\in\Rcal_n^+}\log C_n(r^n).
\end{equation}
Then $\bar r_n(\Wcal)$ is the global conditional Shtarkov minimax log-loss regret:
\begin{align}
\bar r_n(\Wcal) &= \inf_{Q_n} \esssup_{r^n\in\Rcal_n^+} \sup_{W\in\Wcal_n(r^n)} \max_{\substack{e^n\in\A^n\\p_{W,n}(e^n|r^n)>0}} \log\frac{p_{W,n}(e^n|r^n)}{Q_n(e^n|r^n)},
\label{eq:shtarkovglobal}
\end{align}
where the infimum is over measurable conditional pmfs on $\A^n$ given $\Rcal_n$, and the logarithmic ratio is $+\infty$ when $Q_n(e^n|r^n)=0<p_{W,n}(e^n|r^n)$. The infimum is attained by $Q_n^{\rm peNML}$. Moreover,
\begin{equation}
0 \le \bar r_n(\Wcal)- \mathfrak r_n^\star(\Wcal) \le \log\mathsf H_{a^n}.
\label{eq:shtbridgeglobal}
\end{equation}
The order induced by $Q_n^{\rm peNML}$ is a minimax order in Theorem~\ref{thm:exactregret}.
\end{theorem}

\begin{proof}
Fix $r^n\in\Rcal_n^+$ and write $m_{(j)}=m_{n,(j)}(r^n)$ and $K=\max_{1\le j\le a^n}j\,m_{(j)}$. Since $m_{(j)}\le K/j$,
\begin{equation}
C_n(r^n) = \sum_{j=1}^{a^n}m_{(j)} \le  K\mathsf H_{a^n}.
\end{equation}
Conversely, for every $j$,
\begin{equation}
j\,m_{(j)} \le \sum_{i=1}^j m_{(i)} \le C_n(r^n).
\end{equation}
Using \eqref{eq:rhoobs} gives \eqref{eq:shtbridgeobs}.
For any conditional pmf $Q_n$ and fixed $r^n\in\Rcal_n^+$,
\begin{align}
\sup_{W\in\Wcal_n(r^n)} \max_{\substack{e^n\in\A^n\\p_{W,n}(e^n|r^n)>0}} \frac{p_{W,n}(e^n|r^n)}{Q_n(e^n|r^n)} = \max_{\substack{e^n\in\A^n\\m_n(e^n|r^n)>0}} \frac{m_n(e^n|r^n)}{Q_n(e^n|r^n)},
\label{eq:pointwiseloglossratio}
\end{align}
with either side interpreted as $+\infty$ if a denominator vanishes while the corresponding numerator is positive. Moreover,
\begin{align}
C_n(r^n) &= \sum_{e^n:m_n(e^n|r^n)>0}m_n(e^n|r^n)\\
&\le \left[ \max_{\substack{e^n\in\A^n\\m_n(e^n|r^n)>0}} \frac{m_n(e^n|r^n)}{Q_n(e^n|r^n)} \right]
\sum_{e^n:m_n(e^n|r^n)>0}Q_n(e^n|r^n). 
\end{align}
Hence, by \eqref{eq:pointwiseloglossratio}, the pointwise worst-case log-loss regret is at least $\log C_n(r^n)$. Equality is attained by $Q_n^{\rm peNML}$, for which $m_n(e^n|r^n)/Q_n^{\rm peNML}(e^n|r^n)=C_n(r^n)$ whenever $m_n(e^n|r^n)>0$. Taking the $\mu_n$-essential supremum and then the infimum over measurable conditional pmfs proves \eqref{eq:shtarkovglobal}.
Finally, Theorem~\ref{thm:exactregret} gives $\mathfrak r_n^\star(\Wcal) =\esssup_{r^n\in\Rcal_n^+} \mathfrak r_n^\star(r^n)$. Taking essential suprema in \eqref{eq:shtbridgeobs} gives \eqref{eq:shtbridgeglobal}. Since $Q_n^{\rm peNML}(\cdot|r^n)$ is proportional to $m_n(\cdot|r^n)$, its induced order is a nonincreasing-envelope order and is therefore minimax by Theorem~\ref{thm:exactregret}.
\end{proof}

Since $\log\mathsf H_{a^n}=\mathcal O(\log n)=o(n)$, \eqref{eq:shtbridgeglobal} implies
\begin{equation}
\bar r_n(\Wcal)=o(n)
\quad\Longleftrightarrow\quad
\mathfrak r_n^\star(\Wcal)=o(n).
\end{equation}

\begin{proposition}[Posterior log-loss regret and induced rank regret]
\label{prop:domination}
Let $Q_n$ be a measurable conditional pmf on $\A^n$ given $\Rcal_n$, and let $\pi_n^Q$ order corrections by nonincreasing $Q_n(\cdot|r^n)$ with the prescribed deterministic tie-breaking rule. Define the worst-case posterior log-loss regret
\begin{align}
R_n(Q_n,\Wcal) = \esssup_{r^n\in\Rcal_n^+} \sup_{W\in\Wcal_n(r^n)}
\max_{\substack{e^n\in\A^n\\ p_{W,n}(e^n|r^n)>0}}
\log \frac{p_{W,n}(e^n|r^n)}{Q_n(e^n|r^n)},
\label{eq:domreg}
\end{align}
where the logarithmic ratio is $+\infty$ when $Q_n(e^n|r^n)=0<p_{W,n}(e^n|r^n)$. Then
\begin{equation}
\mathfrak r_n(\pi_n^Q,\Wcal) \le R_n(Q_n,\Wcal).
\label{eq:domrank}
\end{equation}
In particular, \eqref{eq:shtarkovglobal} gives $\inf_{Q_n}R_n(Q_n,\Wcal)=\bar r_n(\Wcal)$.
\end{proposition}

\begin{proof}
Fix $r^n\in\Rcal_n^+$, $W\in\Wcal_n(r^n)$, and $e^n\in\A^n$ with $p_{W,n}(e^n|r^n)>0$. If $Q_n(e^n|r^n)=0$, the required pointwise inequality holds in the extended-real sense. Otherwise, the $G_{\pi_n^Q}(e^n|r^n)$ corrections of rank no larger than that of $e^n$ each have $Q_n$-probability at least $Q_n(e^n|r^n)$. Therefore, $G_{\pi_n^Q}(e^n|r^n)Q_n(e^n|r^n)\le1$, and hence
\begin{equation}
G_{\pi_n^Q}(e^n|r^n)p_{W,n}(e^n|r^n) \le \frac{p_{W,n}(e^n|r^n)}{Q_n(e^n|r^n)}.
\end{equation}
Taking logarithms and then the defining suprema proves \eqref{eq:domrank}.
\end{proof}

\begin{proposition}[Conditional cross-entropy bound for expected logarithmic query rank]
\label{prop:crossentropy}
Fix $W$ and let $Q_n$ be a measurable conditional pmf on $\A^n$ given $\Rcal_n$. Under the auxiliary uniform-input distribution associated with $W$, define
\begin{equation}
\mathcal D_n(W\Vert Q_n) = \E_W\!\left[ \log \frac{p_{W,n}(E_U^n|R^n)}{Q_n(E_U^n|R^n)} \right]
\in[0,+\infty],
\label{eq:conditionalD}
\end{equation}
where the expectation is $+\infty$ when $Q_n(E_U^n|R^n)=0$ on an event of positive probability. This quantity is the conditional Kullback--Leibler divergence from the correction posterior to $Q_n$, averaged over $R^n$. Then
\begin{align}
\E_W\!\left[ \log G_{\pi_n^Q}(E_U^n|R^n) \right]
&\le \E_W\!\left[ -\log Q_n(E_U^n|R^n) \right]\\
&= H_W(E_U^n|R^n) + \mathcal D_n(W\Vert Q_n).
\label{eq:crossentropy}
\end{align}
Equivalently,
\begin{align}
\E_W\!\left[ \log\!\left( G_{\pi_n^Q}(E_U^n|R^n) p_{W,n}(E_U^n|R^n) \right) \right] \le
\mathcal D_n(W\Vert Q_n).
\label{eq:avgregret}
\end{align}
\end{proposition}

\begin{proof}
The nonincreasing $Q_n$ order satisfies
\begin{equation}
G_{\pi_n^Q}(e^n|r^n)Q_n(e^n|r^n)\le1
\end{equation}
whenever $Q_n(e^n|r^n)>0$; when $Q_n(e^n|r^n)=0$, the corresponding logarithmic inequality holds in the extended-real sense. Hence
\begin{equation}
\log G_{\pi_n^Q}(E_U^n|R^n) \le -\log Q_n(E_U^n|R^n)
\end{equation}
almost surely. Taking expectations gives the first inequality in \eqref{eq:crossentropy}. The second equality follows from \eqref{eq:conditionalD} and $H_W(E_U^n|R^n) = \E_W[-\log p_{W,n}(E_U^n|R^n)]$. Subtracting this entropy identity gives \eqref{eq:avgregret}.
\end{proof}

\subsection{Conditional-assignment constructions}
\label{subsec:constructQ}

The preceding bounds apply to several conditional assignments $Q_n$.

\emph{- Posterior-envelope NML.}
The assignment in \eqref{eq:penml} induces the nonincreasing posterior-envelope order and therefore attains the finite-block minimax value in Theorem~\ref{thm:exactregret}. The conditional Shtarkov minimax log-loss regret exceeds the minimax posterior-normalized rank regret by at most $\log\mathsf H_{a^n}$ by Theorem~\ref{thm:shtbridge}. The blockwise definition in \eqref{eq:penml} does not provide a recursive procedure for generating the query order without evaluating the required posterior-envelope values.

\emph{- Finite mixtures.}
For a finite channel class, Corollary~\ref{cor:mixture} gives a posterior mixture satisfying a blocklength-independent bound on $R_n(Q_n,\Wcal)$.

\emph{- Posterior mixtures.}
For a parametric family $\{W_\theta\}_{\theta\in\Theta}$ whose conditional posteriors are defined on a common measurable observation domain, write $p_{\theta,n}=p_{W_\theta,n}$. Assume that $(\theta,r^n)\mapsto p_{\theta,n}(e^n|r^n)$ is measurable for every $e^n\in\A^n$. Let $\Pi$ be a probability measure on $\Theta$ and, on this common domain, define $Q_n^\Pi(e^n|r^n) = \int_\Theta p_{\theta,n}(e^n|r^n)\,\Pi(\mathrm d\theta)$. Extend $Q_n^\Pi$ by a fixed conditional pmf outside the common domain. This is a fixed-weight mixture of conditional posteriors and need not equal the posterior induced by the corresponding mixture of channels. The regret and expected-rank bounds in Propositions~\ref{prop:domination} and~\ref{prop:crossentropy} apply directly to $Q_n^\Pi$.

\emph{- Parametric conditional models.}
Let ${Q_{n,\phi}}$ be a family of measurable conditional pmfs. For fixed $W$, suppose a minimizer of $\E_W[-\log Q_{n,\phi}(E_U^n\mid R^n)]$ over $\phi$ exists. By Proposition~\ref{prop:crossentropy}, any such minimizer minimizes the corresponding cross-entropy upper bound on $\E_W[\log G_{\pi_n^{Q_{n,\phi}}}(E_U^n\mid R^n)]$. This does not imply that the same minimizer minimizes the expected logarithmic rank.

\emph{- Plug-in channel estimates.}
Under the joint measurability condition stated above, suppose $\hat\theta_n=\hat\theta_n(R^n)$ is a measurable estimator and define $Q_n(e^n\mid r^n) =p_{\hat\theta_n(r^n),n}(e^n\mid r^n)$. The pointwise posterior log-ratio is
\begin{equation*}
\log \frac{p_{\theta,n}(e^n\mid r^n)}{p_{\hat\theta_n(r^n),n}(e^n\mid r^n)}.
\end{equation*}
A parameter-estimation error bound does not by itself bound this log-ratio; such a bound additionally requires a regularity condition relating parameter distance to posterior log-ratios.

\begin{corollary}[Mixture posterior for a finite channel class]
\label{cor:mixture}

Let $\Wcal=\{W_1,\ldots,W_K\}$ and choose weights $\alpha_j>0$ with $\sum_{j=1}^K\alpha_j=1$. For $r^n\in\Rcal_n^+$, define
\[
\mathcal I_n(r^n) = \{j\in\{1,\ldots,K\}:W_j\in\Wcal_n(r^n)\},
\qquad A_n(r^n) = \sum_{j\in\mathcal I_n(r^n)}\alpha_j.
\]
Since $r^n\in\Rcal_n^+$, $A_n(r^n)>0$. Define
\begin{equation}
Q_n^{\rm mix}(e^n|r^n) = \frac{1}{A_n(r^n)} \sum_{j\in\mathcal I_n(r^n)}
\alpha_j p_{W_j,n}(e^n|r^n).
\label{eq:mix}
\end{equation}
Because $Z_{W_j,n}$ is measurable and $p_{W_j,n}(\cdot|r^n)$ is measurable on $\{Z_{W_j,n}>0\}$, $A_n$ and $Q_n^{\rm mix}$ are measurable on $\Rcal_n^+$. Let $\pi_n^{\rm mix} \coloneqq \pi_n^{Q_n^{\rm mix}}$. Then, with $\alpha_{\min}=\min_{1\le j\le K}\alpha_j$,
\begin{align}
R_n(Q_n^{\rm mix},\Wcal) &\le \log\frac{1}{\alpha_{\min}},\\
\mathfrak r_n(\pi_n^{\rm mix},\Wcal) &\le \log\frac{1}{\alpha_{\min}}.
\end{align}
For equal weights, both upper bounds become $\log K$.
\end{corollary}

\begin{proof}
Fix $r^n\in\Rcal_n^+$ and $j\in\mathcal I_n(r^n)$. Since $A_n(r^n)\le1$,
\begin{equation}
Q_n^{\rm mix}(e^n|r^n) \ge \frac{\alpha_j}{A_n(r^n)} p_{W_j,n}(e^n|r^n)
\ge \alpha_j p_{W_j,n}(e^n|r^n).
\end{equation}
Hence
\[
\log \frac{p_{W_j,n}(e^n|r^n)}{Q_n^{\rm mix}(e^n|r^n)}
\le \log\frac{1}{\alpha_j} \le \log\frac{1}{\alpha_{\min}}.
\]
Taking the required suprema proves the bound on $R_n(Q_n^{\rm mix},\Wcal)$. Proposition~\ref{prop:domination} gives the corresponding bound on posterior-normalized rank regret.
\end{proof}

For infinite parametric families, the following result gives covering-number bounds.

\begin{theorem}[Covering-number bounds for Shtarkov, posterior, and rank regret]
\label{thm:coverregret}
Let $\Wcal=\{W_\theta\mid\theta\in\Theta\}$, where $(\Theta,d_\Theta)$ is a compact metric space. Assume that, for every $n$ and every $r^n\in\Rcal_n^+$, $\Wcal_n(r^n)=\Wcal$, and that there exists $\Scal_n(r^n)\subseteq\A^n$ such that $\{e^n\in\A^n:p_{\theta,n}(e^n|r^n)>0\} = \Scal_n(r^n)$ for every $\theta\in\Theta$. Suppose that, for each $n$, there exists $\mathcal L_n<\infty$ such that
\begin{equation}
\left| \log p_{\theta,n}(e^n|r^n) - \log p_{\theta',n}(e^n|r^n) \right|
\le \mathcal L_n d_\Theta(\theta,\theta')
\label{eq:postlip}
\end{equation}
for all $\theta,\theta'\in\Theta$, all $r^n\in\Rcal_n^+$, and all $e^n\in\Scal_n(r^n)$.
Let $N(\Theta,d_\Theta,\epsilon)$ denote the covering number of $\Theta$. Then, for every $\epsilon_n>0$,
\begin{equation}
\bar r_n(\Wcal) \le \mathcal L_n\epsilon_n + \log N(\Theta,d_\Theta,\epsilon_n).
\label{eq:coverpenml}
\end{equation}
Moreover, for any $\epsilon_n$-net $\{\theta_1,\ldots,\theta_{K_n}\}$, define
\begin{equation}
Q_n^{\rm grid}(e^n|r^n) = \frac{1}{K_n} \sum_{j=1}^{K_n} p_{\theta_j,n}(e^n|r^n).
\label{eq:gridmix}
\end{equation}
Let $\pi_n^{\rm grid}\coloneqq\pi_n^{Q_n^{\rm grid}}$. Then
\begin{align}
R_n(Q_n^{\rm grid},\Wcal) &\le \mathcal L_n\epsilon_n+\log K_n,
\label{eq:gridregret}\\
\mathfrak r_n(\pi_n^{\rm grid},\Wcal) &\le \mathcal L_n\epsilon_n+\log K_n.
\label{eq:gridrankregret}
\end{align}
\end{theorem}

\begin{proof}
Fix $\epsilon_n>0$ and choose an $\epsilon_n$-net $\{\theta_1,\ldots,\theta_{K_n}\}$ with $K_n=N(\Theta,d_\Theta,\epsilon_n)$. For every $\theta\in\Theta$, choose $\theta_j$ satisfying $d_\Theta(\theta,\theta_j)\le\epsilon_n$. For $r^n\in\Rcal_n^+$ and $e^n\in\Scal_n(r^n)$, \eqref{eq:postlip} gives
\begin{equation}
p_{\theta,n}(e^n|r^n) \le a^{\mathcal L_n\epsilon_n} p_{\theta_j,n}(e^n|r^n).
\end{equation}
Hence
\begin{align}
m_n(e^n|r^n) \le a^{\mathcal L_n\epsilon_n} \max_{1\le j\le K_n} p_{\theta_j,n}(e^n|r^n)
\le a^{\mathcal L_n\epsilon_n} \sum_{j=1}^{K_n} p_{\theta_j,n}(e^n|r^n).
\end{align}
The same inequality holds trivially for $e^n\notin\Scal_n(r^n)$ because all posterior probabilities then vanish. Summing over $e^n$ gives
\begin{equation}
C_n(r^n) \le a^{\mathcal L_n\epsilon_n}K_n.
\end{equation}
Taking logarithms and the $\mu_n$-essential supremum over $r^n\in\Rcal_n^+$ proves \eqref{eq:coverpenml}.
For an arbitrary $\epsilon_n$-net of cardinality $K_n$, the same choice of $\theta_j$ gives
\begin{equation}
Q_n^{\rm grid}(e^n|r^n) \ge \frac{1}{K_n} p_{\theta_j,n}(e^n|r^n)
\ge \frac{a^{- \mathcal L_n\epsilon_n}}{K_n} p_{\theta,n}(e^n|r^n).
\end{equation}
Taking the required suprema proves \eqref{eq:gridregret}; Proposition~\ref{prop:domination} gives \eqref{eq:gridrankregret}.
\end{proof}

\begin{corollary}[Logarithmic regret for finite-dimensional classes]
\label{cor:paramlogregret}
Under the assumptions of Theorem~\ref{thm:coverregret}, suppose in addition that $\Theta\subset\mathbb R^q$, $d_\Theta(\theta,\theta')=\|\theta-\theta'\|$, and, for some $A\ge1$,
\begin{equation}
N(\Theta,d_\Theta,\epsilon) \le \left(\frac{A}{\epsilon}\right)^q, \qquad 0<\epsilon\le1.
\end{equation}
Suppose also that \eqref{eq:postlip} holds with $\mathcal L_n\le nL$. For $\epsilon_n=1/n$, choose an $\epsilon_n$-net with $K_n\le(An)^q$. Then
\begin{equation}
\begin{aligned}
\bar r_n(\Wcal) &\le L+q\log(An),\\
R_n(Q_n^{\rm grid},\Wcal) &\le L+q\log(An),\\
\mathfrak r_n(\pi_n^{\rm grid},\Wcal) &\le L+q\log(An).
\end{aligned}
\label{eq:paramlogregret}
\end{equation}
Consequently, $\bar r_n(\Wcal)=\mathcal O(\log n)$, $R_n(Q_n^{\rm grid},\Wcal)=\mathcal O(\log n)$, and $\mathfrak r_n(\pi_n^{\rm grid},\Wcal)=\mathcal O(\log n)$, and all three quantities are $o(n)$.
\end{corollary}

\begin{remark}
\label{rem:coverregretrestricted}
Theorem~\ref{thm:coverregret} requires the log-Lipschitz condition \eqref{eq:postlip} uniformly over all admissible observations. For the generalized-Gaussian model in Section~\ref{sec:signalexamples}, this condition is not uniform over the unbounded observation space. Theorem~\ref{thm:gguniversal} therefore derives the corresponding covering bound on high-probability observation sets and bounds the probabilities of their complements separately.
\end{remark}

\section{Uniform-Subset Random Coding and First-Order Reliability}
\label{sec:coding}

\subsection{Uniform-subset random codebooks}

Fix $W\in\Wcal$ and $1\le M_n\le a^n$. Draw $\C_n$ uniformly from the $M_n$-subsets of $\A^n$. Conditional on $\C_n$, draw $X^n$ uniformly from $\C_n$, and generate $R^n$ according to $W_n(\cdot|X^n)$. Let $\pi_n(R^n)$ be a deterministic, measurable, codebook-independent observation-dependent ordering, and recall that $D_n=G_{\pi_n}(E^n|R^n)$. Let $\widehat X_{\pi_n}^n$ denote the output of GRAND under this ordering, and let
\[
\Pe(\C_n,W;\pi_n) = \Prb\{\widehat X_{\pi_n}^n\neq X^n\mid \C_n\}
\]
denote its block error probability conditional on $\C_n$, where the probability is over the uniform transmitted codeword and the channel output. By Remark~\ref{rem:uniformbenchmark}, after averaging over the random codebook and the uniform transmitted message, the joint distribution of $(E^n,R^n,D_n)$ equals the auxiliary joint distribution of $(E_U^n,R^n,D_{U,n})$.

\begin{theorem}[Exact finite-block rank--collision identity]
\label{thm:hyper}
For GRAND without abandonment under the uniform-subset ensemble,
\begin{equation}
\E_{\C_n}\!\left[\Pe(\C_n,W;\pi_n)\right] = \E_W\!\left[ 1- \frac{\binom{a^n-D_{U,n}}{M_n-1}}{\binom{a^n-1}{M_n-1}} \right].
\label{eq:hyper}
\end{equation}
Moreover, under the joint uniform-subset experiment,
\begin{equation}
\Prb\!\left\{ \widehat X_{\pi_n}^n=X^n \,\middle|\, X^n,R^n \right\} =
\frac{\binom{a^n-D_n}{M_n-1}}{\binom{a^n-1}{M_n-1}}
\quad \mathrm{a.s.}
\label{eq:hypercond}
\end{equation}
with the convention that $\binom{m}{k}=0$ for $m<k$.
\end{theorem}

\begin{proof}
Conditional on $X^n=x^n$, the other $M_n-1$ codewords form a uniformly distributed $(M_n-1)$-subset of $\A^n\setminus\{x^n\}$. Since $R^n$ is conditionally independent of $\C_n$ given $X^n$, the same conditional distribution holds given $(X^n,R^n)=(x^n,r^n)$.
Let $e^n=x^n\ominusA h_n(r^n)$. The map $c^n\mapsto c^n\ominusA h_n(r^n)$ is a bijection from
$\A^n\setminus\{x^n\}$ onto $\A^n\setminus\{e^n\}$. Hence the corrections corresponding to the other codewords form a uniformly distributed $(M_n-1)$-subset of the $a^n-1$ corrections other than $e^n$.
If $D_n=d$, exactly $a^n-d$ corrections occur after $e^n$ in the query order. The decoder returns $X^n$ if and only if all $M_n-1$ corrections corresponding to the other codewords occur after $e^n$. This proves \eqref{eq:hypercond}. Averaging over $(X^n,R^n)$ and using Remark~\ref{rem:uniformbenchmark} gives \eqref{eq:hyper}.
\end{proof}

Conditionally on $D_n=d$, the number of competing-codeword corrections among the $d-1$ preceding corrections has a hypergeometric distribution with population size $a^n-1$, $d-1$ marked elements, and sample size $M_n-1$. Hence \eqref{eq:hypercond} is the corresponding zero collision probability, and \eqref{eq:hyper} follows by averaging its complement.

\begin{corollary}[Rank union bound]
\label{cor:rankunion}
Under the uniform-subset ensemble,
\begin{equation}
\E_{\C_n} \left[\Pe(\C_n,W;\pi_n)\right] \le \E_W\!\left[ \min\!\left\{ 1, \frac{(M_n-1)(D_{U,n}-1)}{a^n-1} \right\} \right].
\label{eq:rankunion}
\end{equation}
If $\pi_n=\pi_n^L$ is induced by a conditionally Kraft-admissible codelength $L_n$, then
\begin{equation}
\E_{\C_n}\!\left[\Pe(\C_n,W;\pi_n^L)\right] \le \E_W\!\left[ \min\!\left\{
1, \frac{(M_n-1) \left(a^{L_n(E_U^n|R^n)}-1\right)}{a^n-1} \right\} \right].
\label{eq:lengthunion}
\end{equation}
\end{corollary}

\begin{proof}
For a realized correction of rank $d$, an error occurs if at least one of the $M_n-1$ other-codeword corrections lies among the $d-1$ preceding corrections. Conditional on the transmitted word and observation, each of these $d-1$ corrections belongs to the uniformly distributed $(M_n-1)$-subset of the remaining $a^n-1$ corrections with probability $(M_n-1)/(a^n-1)$. The union bound, followed by Remark~\ref{rem:uniformbenchmark}, gives \eqref{eq:rankunion}.
If $\pi_n=\pi_n^L$, Lemma~\ref{lem:kraftrank} gives $D_{U,n} \le a^{L_n(E_U^n|R^n)}$. Since the function inside the minimum in \eqref{eq:rankunion} is nondecreasing in the rank, \eqref{eq:lengthunion} follows.
\end{proof}

\begin{theorem}[First-order achievability for a fixed channel]
\label{thm:ach}
Fix $W$ satisfying Assumption~\ref{assump:aep}, and let $\{L_n\}$ be a sequence of conditionally Kraft-admissible codelengths with induced orders $\{\pi_n^L\}$. Suppose
\begin{equation}
\frac{1}{n}L_n(E_U^n|R^n) \xrightarrow{P} h(W)
\quad\text{under }\Prb_W.
\label{eq:lengthh}
\end{equation}
For a fixed $R_{\rm c}\in[0,1]$, let $M_n=\left\lfloor a^{nR_{\rm c}}\right\rfloor$. If $R_{\rm c}<1-h(W)$, then
\begin{equation}
\E_{\C_n} \left[ \Pe(\C_n,W;\pi_n^L) \right] \longrightarrow 0.
\end{equation}
\end{theorem}

\begin{proof}
Choose $\delta>0$ such that $R_{\rm c}+h(W)+\delta<1$, and define $\mathcal A_n = \left\{ L_n(E_U^n|R^n) \le n\bigl(h(W)+\delta\bigr) \right\}$. By \eqref{eq:lengthunion},
\begin{align}
\E_{\C_n} \left[ \Pe(\C_n,W;\pi_n^L) \right] &\le
\Prb_W(\mathcal A_n^c) + \frac{(M_n-1) \left(a^{n(h(W)+\delta)}-1\right)}{a^n-1}.
\end{align}
The first term converges to zero by \eqref{eq:lengthh}. Moreover,
\begin{align}
\frac{(M_n-1) \left(a^{n(h(W)+\delta)}-1\right)}{a^n-1}
&\le \frac{a^{nR_{\rm c}}a^{n(h(W)+\delta)}}{a^n-1} \nonumber \\
&= \frac{a^{-n(1-R_{\rm c}-h(W)-\delta)}}{1-a^{-n}}
\longrightarrow0.
\end{align}
\end{proof}

\begin{corollary}[Memoryless uniform-input information rate]
\label{cor:uniformI}
Consider the stationary memoryless setting of Section~\ref{sec:firstorder} with uniform input and a symbolwise reference decision, and fix $W$. Let $\{Q_n\}$ be posterior-universal at $W$, and let $\{\pi_n^Q\}$ denote the induced orders. For every fixed $0\le R_{\rm c} < I_{\mathrm U,W}(X;R) = 1-H_{\mathrm U,W}(X|R)$, set $M_n=\lfloor a^{nR_{\rm c}}\rfloor$. Under the corresponding uniform-subset ensemble,
\begin{equation}
\E_{\C_n}\!\left[\Pe(\C_n,W;\pi_n^Q)\right] \longrightarrow 0.
\end{equation}
Moreover, define $C(W) \coloneqq \sup_{P_X} I_{P_X,W}(X;R)$, where the supremum is over all pmfs on $\A$. Then $I_{\mathrm U,W}(X;R)\le C(W)$, with equality if and only if the uniform input attains the supremum.
\end{corollary}

\begin{proof}
By \eqref{eq:hwmemoryless}, the memoryless specialization satisfies
Assumption~\ref{assump:aep} with $h(W)=H_{\mathrm U,W}(X|R)$. Posterior universality and \eqref{eq:Qlengthconv} therefore imply
\begin{equation}
\frac{1}{n}\left[-\log Q_n(E_U^n|R^n)\right]
\xrightarrow{P} H_{\mathrm U,W}(X|R)
\quad\text{under }\Prb_W.
\end{equation}
With $L_n=-\log Q_n$, the sequence $\{L_n\}$ is conditionally Kraft-admissible, and, under the prescribed deterministic tie-breaking rule, $\pi_n^L=\pi_n^Q$. Since $R_{\rm c} < 1-H_{\mathrm U,W}(X|R) = 1-h(W)$, Theorem~\ref{thm:ach} gives the stated convergence. The inequality $I_{\mathrm U,W}(X;R)\le C(W)$ follows directly from the definition of $C(W)$, with equality precisely when the uniform input attains the supremum.
\end{proof}

\section{Matched Positive Rank Moments and Error Exponents}
\label{sec:matchedmoments}

Convergence in probability of the normalized logarithmic query rank is insufficient to determine the fixed-rate ensemble-average error exponent, as shown by Example~\ref{ex:firstordernotexponent}. Under the uniform-subset ensemble, the exact finite-block rank--collision identity makes the error probability depend on the distribution of the realized rank, including its right tail. We therefore study positive rank moments. For the matched nonincreasing-posterior order, finite-alphabet guessing inequalities determine the limiting positive rank-moment exponent. Under the differentiability and nondegeneracy conditions specified below, exposed-point right-tail bounds yield the matched uniform-subset error exponent. No complete rank large-deviation principle is assumed.

For a deterministic, measurable, codebook-independent order sequence $\boldsymbol\pi=\{\pi_n\}$ and channel $W$, define
\begin{equation}
\overline P_{e,\boldsymbol\pi}^{(n)}(W) \coloneqq \E_{\C_n}\!\left[
\Pe(\C_n,W;\pi_n) \right]
\label{eq:ensembleerrornotation}
\end{equation}
as the uniform-subset ensemble-average block error probability. For a fixed normalized code rate $R_{\rm c}\in[0,1]$ and a code-size sequence $\{M_n\}$ satisfying $\log_a M_n=nR_{\rm c}+o(n)$, define, with $-\log_a0=+\infty$,
\begin{align}
\underline E_{\boldsymbol\pi}(R_{\rm c};W) & \coloneqq 
\liminf_{n\to\infty} -\frac{1}{n} \log_a \overline P_{e,\boldsymbol\pi}^{(n)}(W), \notag\\
\overline E_{\boldsymbol\pi}(R_{\rm c};W) &  \coloneqq 
\limsup_{n\to\infty} -\frac{1}{n} \log_a \overline P_{e,\boldsymbol\pi}^{(n)}(W).
\label{eq:errorexponentdefinitions}
\end{align}
When these quantities are equal, denote their common value by $E_{\boldsymbol\pi}(R_{\rm c};W)$.

\subsection{Matched positive rank moments}
\label{subsec:matchedmoments}

We use the following dominated standard-Borel extension of Arimoto conditional R\'enyi entropy \cite{arimoto1977}.

\begin{definition}[Arimoto conditional R\'enyi entropy]
\label{def:arimotogeneral}
Let $X$ take values in a finite alphabet $\mathcal X$, and let $S$ take values in a standard Borel space. Suppose $P_{XS}$ is dominated by counting measure on $\mathcal X$ times a $\sigma$-finite measure $\mu$. For each $x\in \mathcal X$, define
\begin{equation}
p_x(s) \coloneqq \frac{\mathrm d P_{X,S}(x,\cdot)}{\mathrm d \mu}(s).
\end{equation}
For $\alpha\in(0,1)$, define
\begin{equation}
H_\alpha^{\rm A}(X|S) \coloneqq \frac{\alpha}{1-\alpha} \log_a \int
\left( \sum_{x\in \mathcal X}p_x(s)^\alpha \right)^{1/\alpha} \mu(\mathrm ds).
\label{eq:arimotogeneral}
\end{equation}
This quantity is independent of the choice of dominating measure.
\end{definition}

\begin{lemma}[Finite-block conditional guessing bounds]
\label{lem:continuousArikan}
Let $X$ take values in a finite alphabet $\mathcal X$ with $|\mathcal X|=M$, and let $S$ take values in a standard Borel space. Suppose the joint distribution is dominated as in Definition~\ref{def:arimotogeneral}. Let $G^{\rm m}(x|s)$ denote the rank of $x$ under a nonincreasing ordering of $P_{X|S}(\cdot|s)$, with the prescribed deterministic tie-breaking rule. For $\rho>0$, let $\alpha=\frac{1}{1+\rho}$ and define
\begin{equation*}
\mathcal J_\alpha \coloneqq \int \left( \sum_{x\in\mathcal X}p_x(s)^\alpha \right)^{1/\alpha}
\mu(\mathrm d s).
\end{equation*}
Then
\begin{equation}
(1+\ln M)^{-\rho}\mathcal J_\alpha \le \E \left[ \bigl(G^{\rm m}(X|S)\bigr)^\rho \right]
\le \mathcal J_\alpha.
\label{eq:continuousArikanSandwich}
\end{equation}
where $\ln$ denotes the natural logarithm.
\end{lemma}

\begin{proof}
Define $p_S(s)=\sum_{x\in\mathcal X}p_x(s)$. For $\mu$-almost every $s$ with $p_S(s)>0$, define
\begin{equation}
P_{X|S}(x|s)=\frac{p_x(s)}{p_S(s)}, \qquad x\in\mathcal X.
\end{equation}
On $\{s:p_S(s)=0\}$, define the conditional pmf and its induced order arbitrarily. For brevity, set
\begin{equation*}
J_\alpha(s) \coloneqq \left( \sum_{x\in\mathcal X} P_{X|S}(x|s)^\alpha \right)^{1/\alpha}.
\end{equation*}
Applying Arikan's finite-alphabet guessing inequalities \cite{arikan1996} to this pmf gives
\begin{equation}
(1+\ln M)^{-\rho}J_\alpha(s) \le \E \left[ \bigl(G^{\rm m}(X|S)\bigr)^\rho \,\middle|\, S=s
\right] \le J_\alpha(s).
\end{equation}
Multiplying by $p_S(s)$ and integrating with respect to $\mu$ gives \eqref{eq:continuousArikanSandwich}, since
\begin{align}
p_S(s) \left( \sum_{x\in\mathcal X} \left( \frac{p_x(s)}{p_S(s)} \right)^\alpha \right)^{1/\alpha}  =
\left( \sum_{x\in\mathcal X} p_x(s)^\alpha \right)^{1/\alpha}
\end{align}
whenever $p_S(s)>0$. On the set where $p_S(s)=0$, all $p_x(s)$ vanish $\mu$-almost everywhere, so this set contributes zero to both sides after integration.
\end{proof}

\begin{lemma}[Continuity of Arimoto conditional entropy at order one]
\label{lem:arimotoalphaone}
Under the hypotheses of Definition~\ref{def:arimotogeneral}, with
$|\mathcal X|=a<\infty$,
\begin{equation}
\lim_{\alpha\uparrow1}H_\alpha^{\rm A}(X|S)=H(X|S).
\label{eq:arimotoalphaone}
\end{equation}
Moreover, if $\alpha_\rho=(1+\rho)^{-1}$ and $\Lambda(\rho)=\rho H_{\alpha_\rho}^{\rm A}(X|S)$ for $\rho>0$, with $\Lambda(0)=0$, then $\Lambda'(0^+)=H(X|S)$.
\end{lemma}

\begin{proof}
Let $p_S(s)=\sum_x p_x(s)$. For $p_S(s)>0$, define $q_x(s)=p_x(s)/p_S(s)$, and define $q(\cdot|s)$ arbitrarily on $\{p_S=0\}$. For $\rho>0$, let $\alpha_\rho=(1+\rho)^{-1}$ and $A(\rho) \coloneqq \int \left(\sum_x p_x(s)^{\alpha_\rho}\right)^{1/\alpha_\rho} \mu(\mathrm ds)$. Then $A(\rho) = \int p_S(s)\, a^{\rho H_{\alpha_\rho}(q(\cdot|s))} \mu(\mathrm ds)$, where $H_\alpha(q)$ denotes the finite-alphabet R\'enyi entropy in base $a$. For $p_S(s)>0$,
\[
H_{\alpha_\rho}(q(\cdot|s)) \longrightarrow H(q(\cdot|s))
\qquad \text{as }\rho\downarrow0,
\]
and $0\le H_{\alpha_\rho}(q(\cdot|s))\le1$. Hence
\[
\frac{a^{\rho H_{\alpha_\rho}(q(\cdot|s))}-1}{\rho} \longrightarrow (\ln a)H(q(\cdot|s)).
\]
For any fixed $\rho_0>0$ and $0<\rho\le\rho_0$,
\[
0 \le \frac{a^{\rho H_{\alpha_\rho}(q(\cdot|s))}-1}{\rho}
\le \frac{a^\rho-1}{\rho} \le (\ln a)a^{\rho_0}.
\]
Since $p_S(s)\mu(\mathrm ds)$ is a probability measure, dominated convergence yields
\[
\lim_{\rho\downarrow0} \frac{A(\rho)-1}{\rho}
= (\ln a) \int p_S(s)H(q(\cdot|s))\mu(\mathrm ds) = (\ln a)H(X|S).
\]
Extend $A$ by $A(0)=1$. Since $A(\rho)\to1$ and $\ln A(\rho)\sim A(\rho)-1$ as $\rho\downarrow0$, $\Lambda(\rho)=\log_a A(\rho)$ for $\rho>0$, and $\Lambda(0)=0$, it follows that
\[
\Lambda'(0^+) = \frac{1}{\ln a} \lim_{\rho\downarrow0}\frac{A(\rho)-1}{\rho} = H(X|S).
\]
Finally, $H_{\alpha_\rho}^{\rm A}(X|S) = \frac{\Lambda(\rho)}{\rho}$, so \eqref{eq:arimotoalphaone} follows as $\rho\downarrow0$.
\end{proof}

\begin{theorem}[Matched positive rank-moment exponent]
\label{thm:continuousmatchedguesswork}
Let $(X_i,S_i)_{i\ge1}$ be i.i.d., where $X_i$ takes values in the finite alphabet $\A$ and $S_i$ takes values in a standard Borel space. Suppose the one-letter joint distribution is dominated as in Definition~\ref{def:arimotogeneral}, with $\mathcal X=\A$. Let $G_n^{\rm m}(\cdot|S^n)$ denote the rank induced by a nonincreasing ordering of $P_{X^n|S^n}(\cdot|S^n)$ with the prescribed deterministic tie-breaking rule, and define $G_n^{\rm m} = G_n^{\rm m}(X^n|S^n)$. Then, for every $\rho>0$, the limit
\begin{equation}
\Lambda_{\rm m}(\rho) \coloneqq \lim_{n\to\infty}
\frac{1}{n}\log_a \E\!\left[(G_n^{\rm m})^\rho\right]
\end{equation}
exists and satisfies
\begin{equation}
\Lambda_{\rm m}(\rho) = \rho H_{1/(1+\rho)}^{\rm A}(X|S).
\label{eq:continuousmatchedmoment}
\end{equation}
In particular,
\begin{equation}
\lim_{n\to\infty} \frac{1}{n}\log_a\E[G_n^{\rm m}] = H_{1/2}^{\rm A}(X|S).
\end{equation}
\end{theorem}

\begin{proof}
Fix $\rho>0$ and let $\alpha=(1+\rho)^{-1}$. Since $|\A|=a$ and $\alpha\in(0,1)$,
\begin{equation}
 \left(
 \sum_{x\in\A}p_x(s)^\alpha
 \right)^{1/\alpha}
 \le
 a^{1/\alpha-1}
 \sum_{x\in\A}p_x(s).
\end{equation}
Hence the Arimoto integral is finite and is at most $a^\rho$.
Apply Lemma~\ref{lem:continuousArikan} to $(X^n,S^n)$. The alphabet of $X^n$ has cardinality $a^n$. Under the i.i.d. product distribution, the Arimoto integral factorizes, equivalently $H_\alpha^{\rm A}(X^n|S^n) = nH_\alpha^{\rm A}(X|S)$. Therefore,
\begin{align}
-\frac{\rho}{n}\log_a(1+n\ln a) +\rho H_\alpha^{\rm A}(X|S) \le
\frac{1}{n}\log_a \E \left[(G_n^{\rm m})^\rho\right] \le
\rho H_\alpha^{\rm A}(X|S).
\end{align}
Since $\frac{1}{n}\log_a(1+n\ln a)\longrightarrow0$, the squeeze theorem gives \eqref{eq:continuousmatchedmoment}.
\end{proof}

\begin{remark}[Classical conditional-guesswork relation]
The relation between positive moments of matched guesswork and Arimoto conditional R\'enyi entropy is classical \cite{arikan1996}. Theorem~\ref{thm:continuousmatchedguesswork} gives the corresponding form for dominated joint distributions with standard-Borel side information.
\end{remark}

\begin{proposition}[Matched rank-moment spectrum and Gallager's $E_0$ function]
\label{prop:gallageridentity}
Let $W$ be a one-letter channel with input alphabet $\A$, $|\A|=a$, and standard-Borel output $S$. Suppose $W(\mathrm ds|x)=w(s|x)\mu(\mathrm ds)$ for a $\sigma$-finite measure $\mu$, and let $X\sim\operatorname{Unif}(\A)$. Let $(X_i,S_i)_{i\ge1}$ be i.i.d. copies of this one-letter law, and let $\Lambda_{\rm m}$ denote the corresponding matched rank-moment spectrum in Theorem~\ref{thm:continuousmatchedguesswork}, extended by $\Lambda_{\rm m}(0)=0$. For $\rho\ge0$, let $\alpha=(1+\rho)^{-1}$ and define Gallager's uniform-input $E_0$ function \cite{gallager1968}
\begin{equation}
E_0^{\rm U}(\rho) = -\log_a \int \left[ \frac{1}{a} \sum_{x\in\A}w(s|x)^\alpha \right]^{1/\alpha} \mu(\mathrm d s).
\label{eq:gallagerE0uniform}
\end{equation}
Then
\begin{equation}
\Lambda_{\rm m}(\rho) = \rho-E_0^{\rm U}(\rho), \qquad \rho\ge0.
\label{eq:gallageridentity}
\end{equation}
Consequently, for every $R_{\rm c}\in[0,1]$,
\begin{equation}
\sup_{0\le\rho\le1} \left\{ \rho(1-R_{\rm c})-\Lambda_{\rm m}(\rho) \right\} =
\sup_{0\le\rho\le1} \left\{ E_0^{\rm U}(\rho)-\rho R_{\rm c} \right\}.
\label{eq:gallagervariational}
\end{equation}
\end{proposition}

\begin{proof}
For $\rho>0$, uniform input gives $p_x(s)=\frac{1}{a}w(s|x)$. Hence
\begin{align}
\int \left( \sum_{x\in\A}p_x(s)^\alpha \right)^{1/\alpha} \mu(\mathrm d s)
&= a^{-1} \int \left( \sum_{x\in\A}w(s|x)^\alpha \right)^{1/\alpha} \mu(\mathrm d s) \notag\\
&= a^{1/\alpha-1}  \int \left[ \frac{1}{a} \sum_{x\in\A}w(s|x)^\alpha \right]^{1/\alpha}  \mu(\mathrm d s).
\end{align}
Since $1/\alpha-1=\rho$, Theorem~\ref{thm:continuousmatchedguesswork} and \eqref{eq:gallagerE0uniform} give $\Lambda_{\rm m}(\rho) = \rho-E_0^{\rm U}(\rho)$. For $\rho=0$, both sides are zero. Substitution into the first variational expression gives \eqref{eq:gallagervariational}.
\end{proof}

Let $h_n(S^n)$ denote the deterministic, measurable, codebook-independent reference decision. For every $s^n$ at which the uniform-input posterior is defined, the map $x^n\mapsto x^n\ominusA h_n(s^n)$ is a bijection on $\A^n$. Hence the matched correction posterior is obtained from $P_{X^n|S^n}(\cdot|s^n)$ by relabeling its outcomes. It follows that, conditional on $S^n=s^n$, the matched correction rank has the same distribution as $G_n^{\rm m}(X^n|s^n)$. This distribution is unchanged by deterministic tie-breaking within posterior-tie classes.

\begin{theorem}[Matched exposed-point right-tail deviations and uniform-subset error exponent]
\label{thm:continuousmatchedrightldp}
Under the hypotheses of Theorem~\ref{thm:continuousmatchedguesswork}, extend the matched rank-moment spectrum by $\Lambda_{\rm m}(0)=0$. Suppose that, for some $\epsilon>0$, $\Lambda_{\rm m}$ is continuously differentiable on $[0,1+\epsilon]$. Let $h\coloneqq H(X|S)$ and assume $0<h<1$. Define
\[
Z_n = \frac{1}{n}\log_a G_n^{\rm m}(X^n|S^n).
\]
Then $Z_n\xrightarrow{P}h$. For every $\rho\in(0,1]$, define $x_\rho=\Lambda_{\rm m}'(\rho)$, $I_\rho=\rho x_\rho-\Lambda_{\rm m}(\rho)$. Then $h\le x_\rho\le1$, $\rho\in(0,1]$, and
\begin{align}
\lim_{\delta\downarrow0} \liminf_{n\to\infty} \frac{1}{n}\log_a \Prb\{|Z_n-x_\rho|<\delta\} =
\lim_{\delta\downarrow0} \limsup_{n\to\infty} \frac{1}{n}\log_a
\Prb\{|Z_n-x_\rho|<\delta\}
= -I_\rho.
\label{eq:exposedrightlocalld}
\end{align}
Moreover, with $I_+(x) = \sup_{t\ge0}\{tx-\Lambda_{\rm m}(t)\}$, we have $I_+(x_\rho)=I_\rho$.
Suppose in addition that $(X,S)$ is induced by a memoryless channel $W$ with input alphabet $\A$, output density $w(s|x)$ with respect to $\mu$, and uniform input $X\sim\operatorname{Unif}(\A)$. Consider the uniform-subset ensemble with $\log_a M_n=nR_{\rm c}+o(n)$ for a fixed $0<R_{\rm c}<1$; then $M_n \geq 2$ for all sufficiently large $n$. For the matched correction-order sequence $\boldsymbol\pi^{\rm m}$, the ensemble-average error exponent exists and satisfies
\begin{align}
E_{\rm m}(R_{\rm c}) &\coloneqq E_{\boldsymbol\pi^{\rm m}}(R_{\rm c};W) \notag\\
&= \sup_{0\le\rho\le1} \left\{ \rho(1-R_{\rm c})-\Lambda_{\rm m}(\rho) \right\} \notag\\
&= \sup_{0\le\rho\le1} \left\{ E_0^{\rm U}(\rho)-\rho R_{\rm c} \right\} \notag\\
&= E_r^{\rm U}(R_{\rm c}),
\label{eq:matchedmomenterrortheorem}
\end{align}
where $E_r^{\rm U}(R_{\rm c}) \coloneqq \sup_{0\le\rho\le1} \left\{ E_0^{\rm U}(\rho)-\rho R_{\rm c} \right\}$. If $1-R_{\rm c}\le h$, then $E_{\rm m}(R_{\rm c})=0$.
\end{theorem}

\begin{proof}
For $\rho>0$, Theorem~\ref{thm:continuousmatchedguesswork} gives
\begin{equation}
\lim_{n\to\infty} \frac{1}{n}\log_a \E \left[a^{n\rho Z_n}\right] = \Lambda_{\rm m}(\rho).
\label{eq:rightmgflimit}
\end{equation}
The same equality holds at $\rho=0$ because both sides are zero. Define $\Lambda_n(\rho) = \frac{1}{n}\log_a\E\!\left[a^{n\rho Z_n}\right]$, $\rho\ge0$. Because $0\le Z_n\le1$, each $\Lambda_n$ is convex, nondecreasing, and $1$-Lipschitz on $[0,\infty)$. The pointwise limit $\Lambda_{\rm m}$ has the same properties. Moreover, Lemma~\ref{lem:arimotoalphaone} gives $\Lambda_{\rm m}'(0^+)=H(X|S)=h$. Convexity and the $1$-Lipschitz property then give
\[
h\le \Lambda_{\rm m}'(\rho)=x_\rho\le1, \qquad \rho\in(0,1].
\]
The conditional self-information $-\log_a P_{X^n|S^n}(X^n|S^n)$ is a sum of i.i.d. integrable random variables and therefore, after normalization by $n$, converges in probability to $h$. The pointwise bound
\begin{equation}
G_n^{\rm m}(X^n|S^n) \le \frac{1}{P_{X^n|S^n}(X^n|S^n)}
\end{equation}
and the conditional AEP give $\plimsup_{n\to\infty} Z_n\le h$.
For the reverse inequality, fix $0<\delta<h$ and choose
$0<\eta<\delta$. Let $t_{n,\delta}=\left\lceil a^{n(h-\delta)}\right\rceil$ and let $A_n(s^n)$ denote the first $t_{n,\delta}$ sequences under the matched ordering. Define
\[
B_n(s^n) = \left\{ x^n: P_{X^n|S^n}(x^n|s^n) > a^{-n(h-\eta)} \right\}.
\]
Then
\begin{align}
\sum_{x^n\in A_n(s^n)} P_{X^n|S^n}(x^n|s^n) \le \sum_{x^n\in B_n(s^n)} P_{X^n|S^n}(x^n|s^n)
+ t_{n,\delta} a^{-n(h-\eta)}.
\end{align}
Averaging over $S^n$ and using the conditional AEP gives $\Prb\{G_n^{\rm m}\le t_{n,\delta}\}\longrightarrow0$. Hence $\pliminf_{n\to\infty}Z_n\ge h$. Therefore $Z_n\xrightarrow{P}h$.
Fix $\rho\in(0,1]$ and, for $\delta>0$, define $A_{n,\delta} = \{|Z_n-x_\rho|<\delta\}$. Define
\begin{equation}
\frac{\mathrm d \Prb_{n,\rho}}{\mathrm d \Prb} = \frac{a^{n\rho Z_n}}{\E[a^{n\rho Z_n}]}.
\label{eq:ranktiltmeasure}
\end{equation}
For $0<t<\min\{\rho,1+\epsilon-\rho\}$, Markov's inequality and \eqref{eq:rightmgflimit} give
\begin{align}
\limsup_{n\to\infty} \frac{1}{n}\log_a \Prb_{n,\rho}\{Z_n\ge x_\rho+\delta\}
&\le -t(x_\rho+\delta)+ \Lambda_{\rm m}(\rho+t)-\Lambda_{\rm m}(\rho), \notag\\
\limsup_{n\to\infty} \frac{1}{n}\log_a \Prb_{n,\rho}\{Z_n\le x_\rho-\delta\}
&\le t(x_\rho-\delta) + \Lambda_{\rm m}(\rho-t)-\Lambda_{\rm m}(\rho).
\label{eq:ranktiltconcentration}
\end{align}
For every fixed $\delta>0$, differentiability of
$\Lambda_{\rm m}$ at $\rho$ makes both bounds strictly negative for
sufficiently small $t>0$. 
Thus $\Prb_{n,\rho}(A_{n,\delta})\to1$, and consequently $\frac{1}{n} \log_a \Prb_{n,\rho}(A_{n,\delta}) \longrightarrow0$. Consequently,
\begin{align}
\Prb(A_{n,\delta}) = \E_{n,\rho} \left[ \id_{A_{n,\delta}} a^{-n\rho Z_n} \right]
\E[a^{n\rho Z_n}] \ge \Prb_{n,\rho}(A_{n,\delta}) a^{-n\rho(x_\rho+\delta)} \E[a^{n\rho Z_n}].
\end{align}
Using \eqref{eq:rightmgflimit} gives
\begin{equation}
\lim_{\delta\downarrow0} \liminf_{n\to\infty} \frac{1}{n}\log_a \Prb(A_{n,\delta}) \ge -I_\rho.
\end{equation}
Conversely, $A_{n,\delta} \subseteq \{Z_n\ge x_\rho-\delta\}$, and Chernoff's inequality with parameter $\rho$ gives
\begin{equation}
\limsup_{n\to\infty} \frac{1}{n}\log_a \Prb(A_{n,\delta}) \le -\rho(x_\rho-\delta) + \Lambda_{\rm m}(\rho).
\end{equation}
Letting $\delta\downarrow0$ proves \eqref{eq:exposedrightlocalld}. Since $\Lambda_{\rm m}$ is convex,
\begin{equation}
I_+(x_\rho) = \rho x_\rho-\Lambda_{\rm m}(\rho) = I_\rho.
\end{equation}
For the coding result, write $P_{e,{\rm m}}^{(n)}=\overline P_{e,\boldsymbol\pi^{\rm m}}^{(n)}(W)$, $N_n=a^n$, $J_n=M_n-1$, $c=1-R_{\rm c}$, and $V_n^{\rm m} =G_n^{\rm m}-1$. By the input-rank/correction-rank equivalence above and Theorem~\ref{thm:hyper}, conditional on $V_n^{\rm m}=k\in\{0,\ldots,N_n-1\}$, the collision probability is
\begin{equation}
f_n(k) = 1- \frac{\binom{N_n-1-k}{J_n}}{\binom{N_n-1}{J_n}},
\label{eq:matchedhyperkernel}
\end{equation}
where the numerator is zero if $N_n-1-k<J_n$. For every $n$ with $M_n\ge2$ and every $k\in\{0,\ldots,N_n-1\}$,
\begin{align}
(1-e^{-1}) \min\left\{ 1,\frac{J_nk}{N_n-1}\right\}
\le f_n(k) \le
\min\left\{ 1,\frac{J_nk}{N_n-1} \right\}.
\label{eq:hyperkerneltwosided}
\end{align}
If $N_n-1-k<J_n$, then $f_n(k)=1$, and the lower bound in \eqref{eq:hyperkerneltwosided} is immediate. Otherwise, $0\le k/(N_n-1-j)\le1$ for $0\le j<J_n$, and
\begin{align}
1-f_n(k) &= \prod_{j=0}^{J_n-1}
\left( 1-\frac{k}{N_n-1-j} \right) \notag\\
& \le \exp\left( -\sum_{j=0}^{J_n-1}\frac{k}{N_n-1-j} \right) \le
\exp\left( -\frac{J_nk}{N_n-1} \right).
\end{align}
The inequality $1-e^{-v}\ge(1-e^{-1})\min\{1,v\}$ for $v\ge0$ gives the lower bound. The upper bound follows from the union bound. Hence, for every integer sequence $k_n\in\{1,\ldots,N_n-1\}$ satisfying
\begin{equation}
\frac{1}{n}\log_a k_n\longrightarrow x>0,
\end{equation}
we have
\begin{equation}
-\frac{1}{n}\log_a f_n(k_n) \longrightarrow (c-x)^+.
\label{eq:collisionexponent}
\end{equation}
For every $0<\rho\le1$, the inequality $\min\{1,u\}\le u^\rho$ and \eqref{eq:hyperkerneltwosided} give
\begin{equation}
P_{e,{\rm m}}^{(n)} \le a^{-n\rho c+o(n)} \, \E[(G_n^{\rm m})^\rho].
\end{equation}
For $\rho=0$, the corresponding exponent lower bound is zero because $\Lambda_{\rm m}(0)=0$ and $P_{e,{\rm m}}^{(n)}\le1$. Therefore,
\begin{equation}
\liminf_{n\to\infty} -\frac{1}{n}\log_aP_{e,{\rm m}}^{(n)} \ge
\sup_{0\le\rho\le1} \{\rho c-\Lambda_{\rm m}(\rho)\}.
\label{eq:matchedexpach}
\end{equation}
Suppose first that $c\le h$. Since $Z_n\xrightarrow{P}h$, for every $0<\delta<h$, $\Prb\{Z_n\ge h-\delta\}\longrightarrow1$. Using the lower bound in \eqref{eq:hyperkerneltwosided} on this event and then letting $\delta\downarrow0$ gives
\begin{equation}
\limsup_{n\to\infty} -\frac{1}{n}\log_aP_{e,{\rm m}}^{(n)} \le0.
\end{equation}
Convexity of $\Lambda_{\rm m}$ and $\Lambda_{\rm m}'(0^+)=h$ imply $\Lambda_{\rm m}(\rho)\ge\rho h$, $\rho\ge0$, so
\begin{equation}
\sup_{0\le\rho\le1} \{\rho c-\Lambda_{\rm m}(\rho)\} = 0.
\end{equation}
Now suppose $c>h$. Let $\rho^\star$ maximize $F(\rho)=\rho c-\Lambda_{\rm m}(\rho)$ over $[0,1]$. Since $F'(0^+)=c-h>0$, every maximizer is positive.
Fix $\rho\in(0,1]$, write $x=x_\rho$, and choose $0<\delta<x$. Define
\begin{equation}
A_{n,\delta}(x)=\{|Z_n-x|<\delta\}, \qquad
k_{n,\delta}(x) = \left\lceil a^{n(x-\delta)}\right\rceil-1.
\end{equation}
On $A_{n,\delta}(x)$, $V_n^{\rm m}=G_n^{\rm m}-1\ge k_{n,\delta}(x)$. Since $f_n(k)$ is nondecreasing in $k$,
\begin{equation}
P_{e,{\rm m}}^{(n)} \ge \Prb\{A_{n,\delta}(x)\} f_n\!\left(k_{n,\delta}(x)\right).
\label{eq:eventwisecollisionlower}
\end{equation}
For each fixed $\delta\in(0,x)$, \eqref{eq:collisionexponent} applied to the deterministic sequence $k_{n,\delta}(x)$ gives
\begin{equation}
-\frac1n\log_a f_n\!\left(k_{n,\delta}(x)\right) \longrightarrow (c-x+\delta)^+.
\end{equation}
Combining this limit with \eqref{eq:exposedrightlocalld} and then letting $\delta\downarrow0$ yields
\begin{equation}
\limsup_{n\to\infty} -\frac1n\log_a P_{e,{\rm m}}^{(n)} \le I_\rho+(c-x_\rho)^+.
\label{eq:eventwisecollisionexponent}
\end{equation}
If $F$ has an interior maximizer, let $\rho^\star\in(0,1)$ be one. Then $x_{\rho^\star}=\Lambda_{\rm m}'(\rho^\star)=c$, and \eqref{eq:eventwisecollisionexponent} gives
\begin{equation}
\limsup_{n\to\infty} -\frac1n\log_a P_{e,{\rm m}}^{(n)} \le
I_{\rho^\star} = \rho^\star c-\Lambda_{\rm m}(\rho^\star).
\end{equation}
If no interior maximizer exists, $\rho^\star=1$. The endpoint optimality condition gives $x_1=\Lambda_{\rm m}'(1)\le c$, and \eqref{eq:eventwisecollisionexponent} gives
\begin{align}
\limsup_{n\to\infty} -\frac1n\log_a P_{e,{\rm m}}^{(n)}
&\le I_1+c-x_1 \notag\\
&= c-\Lambda_{\rm m}(1).
\end{align}
Combining these upper bounds with \eqref{eq:matchedexpach} proves \eqref{eq:matchedmomenterrortheorem}. The final equality follows from Proposition~\ref{prop:gallageridentity}.
\end{proof}

\begin{example}[First-order rank convergence does not determine the error exponent]
\label{ex:firstordernotexponent}
Fix a memoryless channel $W_0$ under uniform input satisfying the hypotheses of Theorem~\ref{thm:continuousmatchedrightldp}, and choose $R_{\rm c}\in(0,1)$ such that $\mathcal E_{\rm m} \coloneqq E_{\boldsymbol\pi^{\rm m}}(R_{\rm c};W_0)>0$. Let $0<\gamma<\mathcal E_{\rm m}$. Let $V_1,V_2,\ldots$ be i.i.d., independent of $\{(X_i,S_i)\}_{i\ge1}$, with $\Prb\{V_i=1\}=a^{-\gamma}$. Define the augmented memoryless channel $\widetilde W_0$ with output $\widetilde S_i=(S_i,V_i)$. Define $\widetilde h_n(s^n,v^n)=h_n(s^n)$. Since $V^n$ is independent of $(X^n,S^n)$, $P_{X^n|\widetilde S^n}(\cdot|s^n,v^n) = P_{X^n|S^n}(\cdot|s^n)$. Thus, under the prescribed tie-breaking rule, the matched correction-rank distribution is unchanged by the augmentation. Consequently, the matched uniform-subset ensemble-average error probability, and hence its exponent $\mathcal E_{\rm m}$, are the same under $\widetilde W_0$ and $W_0$. Let $\mathcal B_n=\{V_1=\cdots=V_n=1\}$, $\Prb(\mathcal B_n)=a^{-n\gamma}$. Define a deterministic, measurable, codebook-independent order $\widetilde\pi_n$ that equals the matched correction order on $\mathcal B_n^{\rm c}$ and uses its reverse permutation on $\mathcal B_n$. Let $D_n^{\rm m}$ and $\widetilde D_n$ denote the realized ranks under the matched and modified orders, respectively. Since the two ranks differ only on $\mathcal B_n$ and $\Prb(\mathcal B_n)\to0$, $\frac1n\log_a\widetilde D_n\xrightarrow{P}h$.

On $\mathcal B_n$, $\widetilde D_n=a^n+1-D_n^{\rm m}$. Moreover, $D_n^{\rm m}$ is independent of $\mathcal B_n$, and its distribution equals that of $G_n^{\rm m}(X^n|S^n)$. Since $n^{-1}\log_aG_n^{\rm m}(X^n|S^n)\xrightarrow{P}h<1$,
\[
\frac{\widetilde D_n}{a^n}\xrightarrow{P}1
\qquad
\text{conditionally on }\mathcal B_n.
\]
For all sufficiently large $n$, let $N_n=a^n$ and $J_n=M_n-1\ge1$. Conditional on $\widetilde D_n$, the no-collision probability satisfies
\[
\frac{\binom{N_n-\widetilde D_n}{J_n}}
{\binom{N_n-1}{J_n}} \le \frac{N_n-\widetilde D_n}{N_n-1},
\]
with the numerator interpreted as zero when $N_n-\widetilde D_n<J_n$. Hence the conditional collision probability on $\mathcal B_n$ converges to one in probability and, since it takes values in $[0,1]$, also in mean. Theorem~\ref{thm:hyper} therefore gives $\overline P_{e,\widetilde{\boldsymbol\pi}}^{(n)} (\widetilde W_0) \ge a^{-n\gamma}(1-o(1))$. Conversely, on $\mathcal B_n^{\rm c}$ the modified order is matched, whereas the conditional error probability on $\mathcal B_n$ is at most one. Hence
$$
\overline P_{e,\widetilde{\boldsymbol\pi}}^{(n)}
(\widetilde W_0)
\le
a^{-n\gamma}
+
\overline P_{e,\boldsymbol\pi^{\rm m}}^{(n)}(W_0).
$$

Since the second term has exponent
$\mathcal E_{\rm m}>\gamma$,

$$
\lim_{n\to\infty}
-\frac1n
\log_a
\overline P_{e,\widetilde{\boldsymbol\pi}}^{(n)}
(\widetilde W_0)
=
\gamma
<
\mathcal E_{\rm m}.
$$

Thus $\widetilde{\boldsymbol\pi}$ has the same first-order normalized log-rank limit as the matched order but a strictly smaller fixed-rate ensemble-average error exponent. No sublinear worst-case posterior-normalized rank-regret condition is imposed on $\widetilde{\boldsymbol\pi}$.
\end{example}

\begin{remark}
\label{rem:subsetiidexponent}
Under the uniform-subset ensemble, the conditional no-collision probability at realized rank $d$ is the hypergeometric quantity in \eqref{eq:hypercond}. Under the i.i.d. uniform random-coding ensemble, the other $M_n-1$ codewords are sampled independently, and the corresponding no-collision probability is $\left(1-\frac{d-1}{a^n}\right)^{M_n-1}$. Thus the two ensembles have different finite-block collision kernels. Theorem~\ref{thm:continuousmatchedrightldp} shows that the matched uniform-subset exponent equals $E_r^{\rm U}(R_{\rm c})$, the classical uniform-input random-coding exponent expression~\cite{gallager1968}. Thus the uniform-subset ensemble has this exponent value although its finite-block collision kernel differs from the i.i.d. kernel above.
\end{remark}

\subsection{Conditional tilt closure and Gallager exponent preservation}
\label{subsec:tiltclosure}

The matched error exponent in Theorem~\ref{thm:continuousmatchedrightldp} depends on
$\Lambda_{\rm m}(\rho)$ for $0\le\rho\le1$. Posterior-regret control yields first-order rank convergence, but Example~\ref{ex:firstordernotexponent} shows that this conclusion alone does not imply preservation of the rank-moment spectrum or the error exponent. We therefore impose a sufficient condition requiring the decoder posterior family to contain the
conditional power tilts associated with the matched moments. Fix a stationary memoryless channel $W_\theta$ in the dominated standard-Borel setting of Section~\ref{sec:firstorder}, under uniform input, and let the reference decision be symbolwise $h_n(r^n)=(h(r_1),\ldots,h(r_n))$. Let $E=X\ominusA h(R)$ denote the one-letter correction, and let $p_\theta(e|r)$ be a fixed measurable version of $P_\theta(E=e|R=r)$. Let $(E_i,R_i)_{i\ge1}$ be i.i.d.\ copies under this one-letter uniform-input law, and write $E^n=(E_1,\ldots,E_n)$ and $R^n=(R_1,\ldots,R_n)$. We use $P_\theta$, $\E_\theta$, and $H_\theta$ for probability, expectation, and
entropy under this law and its product distributions. Let $\mathcal V = \{p_\nu(\cdot|r):\nu\in\mathsf V\}$ be a family of fixed measurable representatives of conditional correction pmfs. For $\rho\ge0$, define $\alpha_\rho = \frac{1}{1+\rho}$ and
\begin{equation}
p_{\theta,\rho}^{\star}(e|r) = \frac{p_\theta(e|r)^{\alpha_\rho}}{Z_{\theta,\rho}(r)},
\qquad Z_{\theta,\rho}(r) = \sum_{u\in\A} p_\theta(u|r)^{\alpha_\rho}.
\label{eq:conditionalpowertilt}
\end{equation}

\begin{definition}[Conditional tilt closure]
\label{def:tiltclosure}
For $\bar\rho>0$, the family $\mathcal V$ is $\bar\rho$-tilt closed at $\theta$ if, for every $\rho\in[0,\bar\rho]$, there exists $\nu_{\theta,\rho}\in\mathsf V$ such that
\begin{equation}
p_{\nu_{\theta,\rho}}(\cdot|r) = p_{\theta,\rho}^{\star}(\cdot|r)
\quad \text{for }P_{\theta,R}\text{-a.e. }r.
\label{eq:tiltclosuredef}
\end{equation}
\end{definition}

Let $H_{\alpha,\theta}^{\rm A}(E|R)$ denote the Arimoto conditional R\'enyi entropy of $(E,R)$ under $P_\theta$, as defined in Definition~\ref{def:arimotogeneral}.
For $n\ge1$, define the product-family posterior envelope
\begin{equation}
m_n^{\mathcal V}(e^n|r^n) = \sup_{\nu\in\mathsf V} \prod_{i=1}^n p_\nu(e_i|r_i)
\label{eq:tiltclosureenvelope}
\end{equation}
and its conditional Shtarkov normalizer
\begin{equation}
C_n^{\mathcal V}(r^n) = \sum_{e^n\in\A^n} m_n^{\mathcal V}(e^n|r^n).
\label{eq:tiltclosureShtarkov}
\end{equation}
Assume that $r^n\mapsto m_n^{\mathcal V}(e^n|r^n)$ is measurable for every $e^n\in\A^n$. Let $G_n^{\mathcal V}(e^n|r^n)$ denote the rank induced by a nonincreasing ordering of
$m_n^{\mathcal V}(\cdot|r^n)$ with the prescribed deterministic tie-breaking rule, and define
\begin{equation}
G_n^{\mathcal V} = G_n^{\mathcal V}(E^n|R^n).
\end{equation}

\begin{theorem}[Rank-moment spectrum preservation under conditional tilt closure]
\label{thm:tiltclosuremoments}
Assume that $\mathcal V$ is $\bar\rho$-tilt closed at $\theta$. Suppose there exist $\kappa>\bar\rho$ and measurable sets $\mathcal T_n$ such that
\begin{equation}
s_n = \sup_{r^n\in\mathcal T_n} \log_a C_n^{\mathcal V}(r^n) = o(n)
\label{eq:tiltclosuresubexpCn}
\end{equation}
and
\begin{equation}
\limsup_{n\to\infty} \frac{1}{n} \log_a P_\theta\{R^n\notin\mathcal T_n\}
\le -\kappa.
\label{eq:tiltclosurebadset}
\end{equation}
For $\rho>0$, define
\begin{equation}
\Lambda_{{\rm m},\theta}(\rho) = \rho H_{1/(1+\rho),\theta}^{\rm A}(E|R),
\end{equation}
and set $\Lambda_{{\rm m},\theta}(0)=0$. Then, for every $\rho\in[0,\bar\rho]$,
\begin{equation}
\lim_{n\to\infty} \frac{1}{n}\log_a \E_\theta
\left[ (G_n^{\mathcal V})^\rho \right] = \Lambda_{{\rm m},\theta}(\rho).
\label{eq:tiltclosuremomentexact}
\end{equation}
\end{theorem}

\begin{proof}
For every $(e^n,r^n)$,
\begin{equation}
G_n^{\mathcal V}(e^n|r^n) m_n^{\mathcal V}(e^n|r^n)
\le C_n^{\mathcal V}(r^n).
\label{eq:tiltclosurerankmass}
\end{equation}
Fix $\rho\in(0,\bar\rho]$. By Definition~\ref{def:tiltclosure}, for $P_{\theta,R^n}$-almost every $r^n$,
\begin{equation}
m_n^{\mathcal V}(e^n|r^n) \ge \prod_{i=1}^n p_{\theta,\rho}^{\star}(e_i|r_i),
\qquad e^n\in\A^n.
\end{equation}
Hence, on $\{R^n\in\mathcal T_n\}$,
\begin{equation}
(G_n^{\mathcal V})^\rho \le a^{\rho s_n} \prod_{i=1}^n p_{\theta,\rho}^{\star}(E_i|R_i)^{-\rho}
\label{eq:tiltclosuremomentupperpointwise}
\end{equation}
almost surely. Since $p_\theta(E_i|R_i)>0$ almost surely for every $i$, $p_{\theta,\rho}^{\star}(E_i|R_i)>0$ almost surely. For $P_{\theta,R}$-almost every $r$, 
\begin{align}
\E_\theta \left[ p_{\theta,\rho}^{\star}(E|r)^{-\rho} \,\middle|\, R=r \right]
&  = \sum_{\substack{e\in\A\\p_\theta(e|r)>0}} p_\theta(e|r) p_{\theta,\rho}^{\star}(e|r)^{-\rho} \notag\\
&  = Z_{\theta,\rho}(r)^\rho \sum_{e\in\A}p_\theta(e|r)^{1-\alpha_\rho\rho}
= Z_{\theta,\rho}(r)^{1+\rho},
\label{eq:tiltclosureoneletteridentity}
\end{align}
because $1-\alpha_\rho\rho=\alpha_\rho$. The i.i.d. product distribution therefore yields
\begin{align}
\E_\theta \left[ (G_n^{\mathcal V})^\rho \id\{R^n\in\mathcal T_n\} \right]
\le a^{\rho s_n} \left( \E_\theta \left[ Z_{\theta,\rho}(R)^{1+\rho} \right] \right)^n.
\end{align}
By Definition~\ref{def:arimotogeneral},
\begin{equation}
\log_a \E_\theta \left[ Z_{\theta,\rho}(R)^{1+\rho} \right]
= \Lambda_{{\rm m},\theta}(\rho).
\end{equation}
Consequently,
\begin{equation}
\limsup_{n\to\infty} \frac{1}{n}\log_a \E_\theta
\left[ (G_n^{\mathcal V})^\rho \id\{R^n\in\mathcal T_n\} \right] \le
\Lambda_{{\rm m},\theta}(\rho).
\label{eq:tiltclosuretypicalupper}
\end{equation}
Since $G_n^{\mathcal V}\le a^n$,
\begin{align}
\E_\theta \left[ (G_n^{\mathcal V})^\rho \id\{R^n\notin\mathcal T_n\} \right]
\le a^{n\rho} P_\theta\{R^n\notin\mathcal T_n\}.
\end{align}
Hence
\begin{equation}
\limsup_{n\to\infty} \frac{1}{n}\log_a \E_\theta \left[ (G_n^{\mathcal V})^\rho \id\{R^n\notin\mathcal T_n\} \right] \le \rho-\kappa<0.
\end{equation}
Since $\Lambda_{{\rm m},\theta}(\rho)\ge0$, the preceding two bounds give
\begin{equation}
\limsup_{n\to\infty} \frac{1}{n}\log_a \E_\theta
\left[ (G_n^{\mathcal V})^\rho \right] \le \Lambda_{{\rm m},\theta}(\rho).
\end{equation}
For $P_{\theta,R^n}$-almost every $r^n$, a nonincreasing ordering of the matched posterior minimizes the conditional $\rho$th rank moment over deterministic permutations. Therefore,
\begin{equation}
\E_\theta \left[ (G_n^{\mathcal V})^\rho \right] \ge
\E_\theta \left[ (G_n^{\rm m})^\rho \right].
\end{equation}
Theorem~\ref{thm:continuousmatchedguesswork} gives
\begin{equation}
\lim_{n\to\infty} \frac{1}{n}\log_a \E_\theta
\left[ (G_n^{\rm m})^\rho \right]
= \Lambda_{{\rm m},\theta}(\rho),
\end{equation}
which gives the matching lower bound. For $\rho=0$, both expectations equal one.
\end{proof}

For the channel $W_\theta$, let $E_{0,\theta}^{\rm U}(\rho)$ denote Gallager's uniform-input $E_0$ function in \eqref{eq:gallagerE0uniform}, evaluated for $W_\theta$. For every $\alpha\in(0,1)$, $H_{\alpha,\theta}^{\rm A}(E|R) = H_{\alpha,\theta}^{\rm A}(X|R)$, because, for $P_{\theta,R}$-almost every $r$, the map $e\mapsto h(r)\oplusA e$ permutes the conditional probabilities. The same argument gives $H_\theta(E|R)=H_\theta(X|R)$. Consequently, Proposition~\ref{prop:gallageridentity} gives
\begin{equation}
\Lambda_{{\rm m},\theta}(\rho) = \rho-E_{0,\theta}^{\rm U}(\rho),
\qquad \rho\ge0.
\label{eq:thetaGallagerIdentity}
\end{equation}
Define
\begin{equation}
E_{r,\theta}^{\rm U}(R_{\rm c}) \coloneqq \sup_{0\le\rho\le1} \left\{ E_{0,\theta}^{\rm U}(\rho)-\rho R_{\rm c} \right\}.
\label{eq:gallagerrandomcodingtheta}
\end{equation}

\begin{theorem}[Gallager-exponent preservation under conditional $1$-tilt closure]
\label{thm:tiltclosureerror}
Assume that $\mathcal V$ is $1$-tilt closed at $\theta$ and that there exist $\kappa>1$ and measurable sets $\mathcal T_n$ satisfying \eqref{eq:tiltclosuresubexpCn} and \eqref{eq:tiltclosurebadset}. Suppose that the matched rank-moment spectrum $\Lambda_{{\rm m},\theta}$ is continuously differentiable on $[0,1+\epsilon]$ for some $\epsilon>0$ and that $0<H_\theta(E|R)<1$. Let $\{M_n\}$ satisfy $\log_a M_n=nR_{\rm c}+o(n)$ for a fixed $0<R_{\rm c}<1$. Let $\boldsymbol\pi^{\mathcal V}$ denote the envelope-order sequence and $\boldsymbol\pi_\theta^{\rm m}$ the matched-order sequence. Then their ensemble-average error exponents exist and satisfy
\begin{align}
E_{\mathcal V}(R_{\rm c};\theta) &= E_{{\rm m},\theta}(R_{\rm c}) \notag\\
&= \sup_{0\le\rho\le1} \left\{ \rho(1-R_{\rm c}) - \Lambda_{{\rm m},\theta}(\rho) \right\} \notag\\
&=  E_{r,\theta}^{\rm U}(R_{\rm c}),
\label{eq:tiltclosureerrorexact}
\end{align}
where $E_{\mathcal V}(R_{\rm c};\theta) = E_{\boldsymbol\pi^{\mathcal V}}(R_{\rm c};W_\theta)$, $E_{{\rm m},\theta}(R_{\rm c}) = E_{\boldsymbol\pi_\theta^{\rm m}}(R_{\rm c};W_\theta)$.
\end{theorem}

\begin{proof}
Write $P_{e,\mathcal V}^{(n)} = \overline P_{e,\boldsymbol\pi^{\mathcal V}}^{(n)}(W_\theta)$, $P_{e,{\rm m}}^{(n)} = \overline P_{e,\boldsymbol\pi_\theta^{\rm m}}^{(n)}(W_\theta)$. Under the uniform-subset ensemble, the ensemble-averaged distribution of $(E^n,R^n)$ is the uniform-input distribution $P_\theta$. For every $0 < \rho\le1$, the rank union bound and $\min\{1,t\}\le t^\rho$ give
\begin{equation}
P_{e,\mathcal V}^{(n)} \le a^{-n\rho(1-R_{\rm c})+o(n)}
\E_\theta \left[ (G_n^{\mathcal V})^\rho \right].
\end{equation}
For $\rho=0$, the corresponding exponent lower bound is zero because $\Lambda_{{\rm m},\theta}(0)=0$ and $P_{e,\mathcal V}^{(n)}\le1$.  Theorem~\ref{thm:tiltclosuremoments}, applied with $\bar\rho=1$, therefore gives
\begin{equation}
\underline E_{\boldsymbol\pi^{\mathcal V}} (R_{\rm c};W_\theta)
\ge \sup_{0\le\rho\le1} \left\{ \rho(1-R_{\rm c}) - \Lambda_{{\rm m},\theta}(\rho) \right\}.
\label{eq:tiltclosureerrorlower}
\end{equation}
For the matched exponent, apply Theorem~\ref{thm:continuousmatchedrightldp} to $W_\theta$ with uniform input $X$ and output $R$. For $P_{\theta,R}$-almost every $r$, the map $x\mapsto x\ominusA h(r)$ is a bijection on $\A$. Hence, conditional on $R^n=r^n$, the matched correction rank has the same distribution as the matched input rank. Moreover, $H_{\alpha,\theta}^{\rm A}(E|R)=H_{\alpha,\theta}^{\rm A}(X|R)$ for every $\alpha\in(0,1)$ and $H_\theta(E|R)=H_\theta(X|R)$. Thus the matched rank-moment spectrum in Theorem~\ref{thm:continuousmatchedrightldp} equals $\Lambda_{{\rm m},\theta}$ defined above. The coding conclusion of that theorem therefore gives
\begin{equation}
E_{{\rm m},\theta}(R_{\rm c}) = \sup_{0\le\rho\le1}
\left\{ \rho(1-R_{\rm c}) - \Lambda_{{\rm m},\theta}(\rho) \right\}
= E_{r,\theta}^{\rm U}(R_{\rm c}).
\label{eq:tiltclosurematchedexponent}
\end{equation}
For every codebook $\C_n$, complete enumeration in matched posterior order returns an ML codeword. Since messages are equiprobable, ML minimizes the block error probability conditional on the fixed codebook. Therefore, $\Pe(\C_n,W_\theta;\pi_n^{\mathcal V}) \ge \Pe(\C_n,W_\theta;\pi_{\theta,n}^{\rm m})$. Averaging over $\C_n$ gives $P_{e,\mathcal V}^{(n)} \ge P_{e,{\rm m}}^{(n)}$, which implies
\begin{equation}
\overline E_{\boldsymbol\pi^{\mathcal V}} (R_{\rm c};W_\theta)
\le E_{{\rm m},\theta}(R_{\rm c}).
\end{equation}
Together with \eqref{eq:tiltclosureerrorlower} and $\underline E_{\boldsymbol\pi^{\mathcal V}} \le \overline E_{\boldsymbol\pi^{\mathcal V}}$, this proves \eqref{eq:tiltclosureerrorexact}.
\end{proof}

\begin{remark}
\label{rem:tiltlimits}
Theorem~\ref{thm:tiltclosureerror} does not require a complete rank large-deviation principle and does not establish lower-tail rank asymptotics. In particular, it does not imply error-exponent preservation for a finite posterior-grid order without a separate exponential-scale
approximation argument.
\end{remark}

\section{Generalized-Gaussian BPSK with Order Reversal and Tilt Closure}
\label{sec:signalexamples}

\subsection{BPSK-AWGN with unknown variance}

Let $\A=\{0,1\}$, $\phi(x)=(-1)^x$, and $R_i=\phi(X_i)+Z_i$, $Z_i\stackrel{\mathrm{i.i.d.}}{\sim}\mathcal N(0,\sigma^2)$. With $h(r)=\id\{r<0\}$, the matched log-likelihood ratio is
\begin{equation}
 \lambda_\sigma(r)
 =
 \ln\frac{w_\sigma(r|0)}{w_\sigma(r|1)}
 =
 \frac{2r}{\sigma^2},
\end{equation}
and the correction-posterior odds satisfy
\begin{equation}
 \frac{p_\sigma(E_i=1|r_i)}
      {p_\sigma(E_i=0|r_i)}
 =
 \exp \left(-|\lambda_\sigma(r_i)|\right).
\end{equation}
Hence the matched order ranks $e^n\in\A^n$ by nondecreasing
\begin{equation}
 \sum_{i=1}^n e_i|\lambda_\sigma(r_i)|
 =
 \frac{2}{\sigma^2}
 \sum_{i=1}^n e_i|r_i|.
\label{eq:sgrandscore}
\end{equation}

The score in \eqref{eq:sgrandscore} induces the matched SGRAND order \cite{solomon2020soft}. If $\sigma^2$ is unknown but constant over the block, the positive factor $2/\sigma^2$ does not affect the weak ordering of the correction patterns. Hence, under the prescribed deterministic tie-breaking rule, the resulting query permutation is independent of $\sigma^2$.

\subsection{Order reversal under unknown generalized-Gaussian shape}

Retain the BPSK map and use the sign reference decision $h(y)=\id\{y<0\}$. Write $Y_i=(-1)^{X_i}+Z_i$, where the additive noise samples are i.i.d.\ with density \cite{song2006ggestimation}
\begin{equation}
f_{\beta,\sigma}(z) = \frac{\beta}{2\sigma\Gamma(1/\beta)}
\exp\!\left[ -\left(\frac{|z|}{\sigma}\right)^\beta \right],
\label{eq:ggnormalized}
\end{equation}
with $\beta\in[1,2]$ and $\sigma>0$, where $\sigma$ is the generalized-Gaussian scale parameter. When $\beta=2$, $Z_i\sim\mathcal N(0,\sigma^2/2)$. For $r\ge0$, define $\ell_\beta(r) = (r+1)^\beta-|r-1|^\beta$. Under the sign reference decision, with $r=|y|$,
\begin{equation}
\ln \frac{p_{\beta,\sigma}(E=0|Y=y)}{p_{\beta,\sigma}(E=1|Y=y)}
= \sigma^{-\beta}\ell_\beta(r).
\end{equation}
Hence, for fixed $(\beta,\sigma)$, the matched correction order is obtained by ordering $e^n\in\{0,1\}^n$ in nondecreasing $\sum_{i=1}^n e_i\ell_\beta(|y_i|)$, with the prescribed deterministic tie-breaking rule, because $\sigma^{-\beta}$ is common to all correction patterns.

Consider $(|y_1|,|y_2|,|y_3|)=(2,0.6,0.6)$, $e_A=(1,0,0)$, and $e_B=(0,1,1)$. For $\beta=1$, $\sum_{i=1}^3(e_A)_i\ell_1(|y_i|) = 2 < 2.4 = \sum_{i=1}^3(e_B)_i\ell_1(|y_i|)$. For $\beta=2$, $\sum_{i=1}^3(e_A)_i\ell_2(|y_i|) = 8 > 4.8 = \sum_{i=1}^3(e_B)_i\ell_2(|y_i|)$. Thus $e_A$ precedes $e_B$ for $\beta=1$, whereas $e_B$ precedes $e_A$ for $\beta=2$. Therefore, changing the shape parameter can reverse the relative order of two correction patterns for the same received block.

\begin{theorem}[Posterior-grid and rank-regret bounds for generalized-Gaussian BPSK]
\label{thm:gguniversal}
Let $\Theta = [\beta_-,\beta_+]\times[\sigma_-,\sigma_+]$, $1\le\beta_-\le\beta_+\le2$, $0<\sigma_-<\sigma_+<\infty$. Let $X_i\stackrel{\mathrm{i.i.d.}}{\sim} \operatorname{Unif}\{0,1\}$ and $Y_i=(-1)^{X_i}+Z_i$, where, for $\theta=(\beta,\sigma)\in\Theta$, the $Z_i$ are i.i.d. with density $f_{\beta,\sigma}$ and are independent of $X^n$. Define $E_i=X_i\oplus h(Y_i)$, $h(y)=\id\{y<0\}$. For $\theta=(\beta,\sigma)\in\Theta$, let $p_\theta(e|y)$ be the fixed measurable conditional-posterior version
\begin{equation}
p_\theta(e|y) = \frac{f_{\beta,\sigma}\!\left(y-(-1)^{e\oplus h(y)}\right)}{
f_{\beta,\sigma}(y-1)+f_{\beta,\sigma}(y+1)},
\qquad e\in\{0,1\},
\label{eq:ggposteriorversion}
\end{equation}
and define $p_{\theta,n}(e^n|y^n) = \prod_{i=1}^n p_\theta(e_i|y_i)$. For every $\kappa>0$, there exists $c_\kappa>0$ such that, with
\begin{equation}
\mathcal T_{n,\kappa} = \left\{ y^n\in\mathbb R^n: \max_{1\le i\le n}|y_i| \le B_{n,\kappa}
\right\}, \qquad
B_{n,\kappa} = 1+c_\kappa n^{1/\beta_-},
\label{eq:ggtypicalset}
\end{equation}
\begin{equation}
\limsup_{n\to\infty} \frac1n \log_2 \sup_{\theta\in\Theta}
P_\theta\{Y^n\notin\mathcal T_{n,\kappa}\} \le-\kappa.
\label{eq:ggbadprob}
\end{equation}
Define $\vartheta(\theta)\coloneqq(\beta,\ln\sigma)$, $\Theta_\eta\coloneqq\vartheta(\Theta)$. For every $y^n\in\mathcal T_{n,\kappa}$,
\begin{equation}
\left| \log_2 p_{\theta,n}(e^n|y^n) - \log_2 p_{\theta',n}(e^n|y^n) \right|
\le \mathcal L_{n,\kappa}^{\mathcal T} \left\| \vartheta(\theta)-\vartheta(\theta') \right\|_2
\label{eq:ggposteriorLip}
\end{equation}
for all $\theta,\theta'\in\Theta$ and $e^n\in\{0,1\}^n$, where
\begin{equation}
\mathcal L_{n,\kappa}^{\mathcal T} \le Cn \left[ 1+ B_{n,\kappa}^{\beta_+} \bigl(1+\ln(2+B_{n,\kappa})\bigr) \right]
\label{eq:ggLip}
\end{equation}
for a finite constant $C$ depending only on $\Theta$.
Consequently, there exist a finite constant $C_{\Theta,\kappa}$ and finite posterior-grid mixtures $Q_{n,\kappa}^{\rm gg}$ such that, for all sufficiently large $n$,
\begin{equation}
\sup_{\substack{ y^n\in\mathcal T_{n,\kappa}\\ \theta\in\Theta,\, e^n\in\{0,1\}^n}}
\log_2 \frac{p_{\theta,n}(e^n|y^n)}{Q_{n,\kappa}^{\rm gg}(e^n|y^n)}
\le C_{\Theta,\kappa}\log_2 n.
\label{eq:ggregret}
\end{equation}
Let $\pi_{n,\kappa}^{\rm gg}(y^n)$ order corrections in nonincreasing $Q_{n,\kappa}^{\rm gg}(\cdot|y^n)$, using the prescribed deterministic tie-breaking rule. Its induced posterior-normalized rank regret on $\mathcal T_{n,\kappa}$ satisfies, for all sufficiently large $n$,
\begin{equation}
\sup_{\substack{ y^n\in\mathcal T_{n,\kappa}\\ \theta\in\Theta,\, e^n\in\{0,1\}^n}}
\log_2\!\left[ G_{\pi_{n,\kappa}^{\rm gg}}(e^n|y^n) p_{\theta,n}(e^n|y^n) \right]
\le C_{\Theta,\kappa}\log_2 n.
\label{eq:ggrankregret}
\end{equation}
Then, for every fixed $\theta\in\Theta$,
\begin{equation}
\frac1n \log_2 G_{\pi_{n,\kappa}^{\rm gg}}(E^n|Y^n) \xrightarrow{P}
H_\theta(E|Y) = H_\theta(X|Y) \quad\text{under }P_\theta.
\label{eq:ggfirstorder}
\end{equation}
\end{theorem}

\begin{proof}
For uniform binary input,
\begin{equation}
P_\theta(X=x|Y=y) = \frac{f_{\beta,\sigma}(y-(-1)^x)}{f_{\beta,\sigma}(y-1)+f_{\beta,\sigma}(y+1)}.
\label{eq:gginputposterior}
\end{equation}
Equation~\eqref{eq:ggposteriorversion} follows because, for each fixed $y$, the map $x\mapsto x\oplus h(y)$ is a permutation of $\{0,1\}$.
Let $\eta=\ln\sigma$ and $t=|z|e^{-\eta}$. Then
\begin{align}
\partial_\eta\ln f_{\beta,e^\eta}(z) &= -1+\beta t^\beta, \notag\\
\partial_\beta\ln f_{\beta,e^\eta}(z)
&= \frac1\beta + \frac{\psi(1/\beta)}{\beta^2} - t^\beta\ln t,
\label{eq:ggdensitygradients}
\end{align}
where $t^\beta\ln t\coloneqq0$ at $t=0$. Compactness of $\Theta_\eta$ implies that, for every $B\ge1$,
\begin{equation}
\sup_{\substack{\vartheta\in\Theta_\eta\\ |y|\le B,\,x\in\{0,1\}}}
\left\| \nabla_\vartheta \ln f_{\beta,\sigma}(y-(-1)^x) \right\|_2
\le C_0 \left[ 1+ B^{\beta_+} \bigl(1+\ln(2+B)\bigr) \right]
\label{eq:ggdensitygradientbound}
\end{equation}
for a finite $C_0$ depending only on $\Theta$. Indeed, $|y-(-1)^x|\le B+1$; the parameter-only terms in \eqref{eq:ggdensitygradients} are uniformly bounded on $\Theta_\eta$; $t^\beta|\ln t|$ is uniformly bounded for $0\le t\le1$; and for $t\ge1$ it is bounded by a constant multiple of $B^{\beta_+}(1+\ln(2+B))$.
The gradient of the logarithm of the denominator in \eqref{eq:gginputposterior} is a convex combination of the two corresponding log-density gradients. Hence the norm of the one-letter log-posterior gradient is at most twice the right-hand side of \eqref{eq:ggdensitygradientbound}. Since $\Theta_\eta = [\beta_-,\beta_+]\times[\ln\sigma_-,\ln\sigma_+]$ is convex, summing over coordinates and applying the mean-value theorem
gives \eqref{eq:ggposteriorLip}--\eqref{eq:ggLip} after conversion to base-two logarithms.
For $u\ge\sigma_+$, let $s=1/\beta\in[1/2,1]$. Since $1\le\beta\le2$, we have $s\in[1/2,1]$. The generalized-Gaussian tail satisfies
\begin{equation}
P_\theta\{|Z|>u\} = \frac{ \Gamma\!\left( s,(u/\sigma)^\beta \right)}{\Gamma(s)}.
\end{equation}
For $x\ge1$ and $s\in[1/2,1]$, $\Gamma(s,x) = \int_x^\infty t^{s-1}e^{-t}\, \mathrm{d} t \le e^{-x}$. Since $\inf_{s\in[1/2,1]}\Gamma(s)>0$, there exists $C_1<\infty$, depending only on $\Theta$, such that, for sufficiently large $B$,
\begin{equation}
\sup_{\theta\in\Theta} P_\theta\{|Y|>B\} \le
C_1 \exp\!\left[ -\left( \frac{B-1}{\sigma_+} \right)^{\beta_-} \right].
\label{eq:gguniformtail}
\end{equation}
The union bound with $B=B_{n,\kappa}$ therefore gives
\begin{equation}
\log_2 \sup_{\theta\in\Theta} P_\theta\{Y^n\notin\mathcal T_{n,\kappa}\} \le
-\frac1{\ln2} \left( \frac{c_\kappa}{\sigma_+} \right)^{\beta_-}n + \mathcal O(\log n).
\end{equation}
Choosing $c_\kappa$ so that $\frac1{\ln2} \left( \frac{c_\kappa}{\sigma_+} \right)^{\beta_-} > \kappa$ proves \eqref{eq:ggbadprob}.
From \eqref{eq:ggLip} and $B_{n,\kappa}=\mathcal O(n^{1/\beta_-})$,
\begin{equation}
\mathcal L_{n,\kappa}^{\mathcal T} = \mathcal O\!\left( n^{1+\beta_+/\beta_-}\log n \right).
\end{equation}
Fix $q>1+\frac{\beta_+}{\beta_-}$ and set $\epsilon_n=n^{-(q+1)}$. For all sufficiently large $n$, $\mathcal L_{n,\kappa}^{\mathcal T}\le n^q$. Since $\Theta_\eta\subset\mathbb R^2$ is compact, it admits an $\epsilon_n$-net $\{\vartheta_{n,j}\}_{j=1}^{K_n}\subset\Theta_\eta$ with $K_n=\mathcal O(\epsilon_n^{-2})$. Let $\theta_{n,j}\in\Theta$ correspond to $\vartheta_{n,j}$ and define
\begin{equation}
Q_{n,\kappa}^{\rm gg}(e^n|y^n) = \frac1{K_n} \sum_{j=1}^{K_n} p_{\theta_{n,j},n}(e^n|y^n).
\label{eq:gggridmixture}
\end{equation}
For every $\theta\in\Theta$, choose $j$ such that $\|\vartheta(\theta)-\vartheta_{n,j}\|_2\le\epsilon_n$. Then, on $\mathcal T_{n,\kappa}$,
\begin{align}
\log_2 \frac{p_{\theta,n}(e^n|y^n)}{Q_{n,\kappa}^{\rm gg}(e^n|y^n)}
\le \mathcal L_{n,\kappa}^{\mathcal T}\epsilon_n + \log_2 K_n 
= \mathcal O(\log n),
\end{align}
uniformly over $\theta$ and $e^n$, proving \eqref{eq:ggregret}. Because $\pi_{n,\kappa}^{\rm gg}$ orders corrections in nonincreasing $Q_{n,\kappa}^{\rm gg}(\cdot|y^n)$,
\[
G_{\pi_{n,\kappa}^{\rm gg}}(e^n|y^n) Q_{n,\kappa}^{\rm gg}(e^n|y^n) \le1.
\]
Combining this inequality with \eqref{eq:ggregret} gives \eqref{eq:ggrankregret}.
Fix $\theta\in\Theta$ and $\epsilon>0$. For all sufficiently large $n$, the right-hand side of \eqref{eq:ggregret} is less than $n\epsilon$. Hence
\begin{align}
P_\theta\!\left\{ \frac1n \log_2 \frac{p_{\theta,n}(E^n|Y^n)}{Q_{n,\kappa}^{\rm gg}(E^n|Y^n)} >\epsilon \right\} \le P_\theta\{Y^n\notin\mathcal T_{n,\kappa}\} \longrightarrow0.
\end{align}
Lemma~\ref{lem:regretnegativetail} controls the negative tail, so $Q_{n,\kappa}^{\rm gg}$ is posterior-universal at $\theta$. Moreover, $\E_\theta[-\log_2p_\theta(E|Y)] = H_\theta(E|Y) \le1$, so the i.i.d. conditional AEP holds. Theorem~\ref{thm:exactrank} therefore gives \eqref{eq:ggfirstorder}. Finally, $x\mapsto x\oplus h(y)$ is bijective for every $y$, so $H_\theta(E|Y)=H_\theta(X|Y)$.
\end{proof}

\begin{remark}
The incomplete-gamma tail bound in the proof of Theorem~\ref{thm:gguniversal} uses $1/\beta\le1$, and therefore requires $\beta_-\ge1$. The theorem does not cover $\beta<1$.
\end{remark}

\subsection{Tilt closure for generalized-Gaussian BPSK}
\label{subsec:ggtiltclosure}

For $y\in\mathbb R$, let $r=|y|$. Under the sign reference decision,
the correction-posterior log-odds satisfy
\begin{equation}
\lambda_{\beta,\sigma}(r) \coloneqq  \ln\frac{p_{\beta,\sigma}(E=0|Y=y)}{p_{\beta,\sigma}(E=1|Y=y)} = \sigma^{-\beta}\ell_\beta(r), \qquad \ell_\beta(r) = (r+1)^\beta-|r-1|^\beta.
\label{eq:gglogit}
\end{equation}

\begin{theorem}[Gallager-exponent preservation for generalized-Gaussian BPSK]
\label{thm:ggtiltcomplete}
Let the prescribed channel parameter set be $\Theta_0 = [\beta_-,\beta_+]\times[\sigma_-,\sigma_+]$, $1\le\beta_-\le\beta_+\le2$, $0<\sigma_-<\sigma_+<\infty$, and let the decoder parameter set be $\widetilde\Theta = [\beta_-,\beta_+]\times[\sigma_-,\widetilde\sigma_+]$, $\widetilde\sigma_+ \ge 2^{1/\beta_-}\sigma_+$. For $\nu=(\beta,\sigma)\in\widetilde\Theta$, let $p_\nu$ denote the conditional correction-posterior kernel induced by the generalized-Gaussian BPSK channel under uniform input and the sign reference decision $h(y)=\id\{y<0\}$. Explicitly,
\begin{equation}
p_\nu(e|y) = \frac{ f_{\beta,\sigma} \bigl(y-(-1)^{e\oplus h(y)}\bigr)}{f_{\beta,\sigma}(y-1)
+ f_{\beta,\sigma}(y+1)},
\qquad e\in\{0,1\}.
\label{eq:ggcorrectionposterior}
\end{equation}
Define $\mathcal V_{\widetilde\Theta} \coloneqq \{p_\nu:\nu\in\widetilde\Theta\}$.
For $n\ge1$, define the product-posterior envelope
\begin{equation}
m_n^{\mathcal V_{\widetilde\Theta}}(e^n|y^n) \coloneqq \sup_{\nu\in\widetilde\Theta}
\prod_{i=1}^n p_\nu(e_i|y_i),
\label{eq:ggenvelope}
\end{equation}
and let $\pi_n^{\rm tc}(y^n)$ order corrections in nonincreasing $m_n^{\mathcal V_{\widetilde\Theta}}(\cdot|y^n)$, using the prescribed deterministic tie-breaking rule. Define $G_n^{\rm tc} \coloneqq G_{\pi_n^{\rm tc}}(E^n|Y^n)$.
For every $\theta\in\Theta_0$, the decoder posterior family $\mathcal V_{\widetilde\Theta}$ is $1$-tilt closed at $\theta$. Moreover, for every $\rho\in[0,1]$,
\begin{equation}
\lim_{n\to\infty} \frac{1}{n} \log_2 \E_\theta \left[ (G_n^{\rm tc})^\rho \right]
= \Lambda_{{\rm m},\theta}(\rho).
\label{eq:ggtiltcompletemoment}
\end{equation}
For every finite $K>0$, the matched spectrum $\Lambda_{{\rm m},\theta}$ is continuously differentiable on $[0,K]$, and $0 < H_\theta(E|Y) = H_\theta(X|Y) < 1$.
Hence, under the uniform-subset code ensemble, for every fixed $0<R_{\rm c}<1$ and every code-size sequence satisfying $\log_2 M_n=nR_{\rm c}+o(n)$, the ensemble-average error exponent $E_{\rm tc}(R_{\rm c};\theta) \coloneqq E_{\boldsymbol\pi^{\rm tc}}(R_{\rm c};W_\theta)$, $\boldsymbol\pi^{\rm tc}\coloneqq\{\pi_n^{\rm tc}\}_{n\ge1}$, exists and satisfies
\begin{align}
E_{\rm tc}(R_{\rm c};\theta) &= E_{{\rm m},\theta}(R_{\rm c}) \notag\\
&= \sup_{0\le\rho\le1} \left\{ \rho(1-R_{\rm c}) - \Lambda_{{\rm m},\theta}(\rho) \right\} \notag\\
&= E_{r,\theta}^{\rm U}(R_{\rm c}).
\label{eq:ggtiltcompleteerror}
\end{align}
\end{theorem}

\begin{proof}
Fix $\theta=(\beta,\sigma)\in\Theta_0$ and $\rho\in[0,1]$, and set $\alpha_\rho=\frac{1}{1+\rho}$. For a binary pmf $(p_0,p_1)$ with positive entries, the normalized $\alpha_\rho$-power transform has odds ratio $\left(\frac{p_0}{p_1}\right)^{\alpha_\rho}$. Hence its log-odds equal $\alpha_\rho$ times the original log-odds. By \eqref{eq:gglogit}, the conditional power tilt $p_{\theta,\rho}^\star(\cdot|y)$ therefore has log-odds $\alpha_\rho \sigma^{-\beta} \ell_\beta(|y|)$. Define
\begin{equation}
\sigma_\rho \coloneqq \sigma\alpha_\rho^{-1/\beta} = \sigma(1+\rho)^{1/\beta}.
\label{eq:ggtiltscale}
\end{equation}
Since $\sigma_\rho^{-\beta} = \alpha_\rho\sigma^{-\beta}$, the two strictly positive binary pmfs have the same odds. Hence $p_{\theta,\rho}^\star(\cdot|y) = p_{(\beta,\sigma_\rho)}(\cdot|y)$, $y\in\mathbb R$. Moreover,
\begin{align}
\sigma_- \le \sigma \le \sigma_\rho \le
\sigma\,2^{1/\beta} \le \sigma_+\,2^{1/\beta_-}
\le \widetilde\sigma_+.
\end{align}
Thus $(\beta,\sigma_\rho)\in\widetilde\Theta$ for every $\rho\in[0,1]$, proving that $\mathcal V_{\widetilde\Theta}$ is $1$-tilt closed at every $\theta\in\Theta_0$.
We next verify measurability of the envelope. For each fixed $(e^n,y^n)$, the map $\nu \longmapsto \prod_{i=1}^n p_\nu(e_i|y_i)$ is continuous on the compact set $\widetilde\Theta$. Let $D\subset\widetilde\Theta$ be a countable dense subset. Continuity in $\nu$ gives
\begin{equation}
m_n^{\mathcal V_{\widetilde\Theta}}(e^n|y^n) = \sup_{\nu\in D} \prod_{i=1}^n p_\nu(e_i|y_i).
\label{eq:ggenvelopemeasurable}
\end{equation}
For every $\nu\in D$, the function of $y^n$ on the right-hand side is measurable. Hence $y^n\mapsto m_n^{\mathcal V_{\widetilde\Theta}}(e^n|y^n)$ is measurable as a countable supremum of measurable functions.

Fix $\kappa>1$. Apply Theorem~\ref{thm:gguniversal} to the compact parameter set $\widetilde\Theta$. It yields measurable sets $\mathcal T_n$ and finite posterior-grid mixtures $Q_n$ such that
\begin{equation}
\limsup_{n\to\infty} \frac{1}{n} \log_2 \sup_{\nu\in\widetilde\Theta}
P_\nu\{Y^n\notin\mathcal T_n\} \le -\kappa
\label{eq:ggtiltbadset}
\end{equation}
and
\begin{equation}
\Gamma_n \coloneqq \sup_{\substack{ y^n\in\mathcal T_n\\ \nu\in\widetilde\Theta,\, e^n\in\{0,1\}^n}} \log_2 \frac{\prod_{i=1}^n p_\nu(e_i|y_i)}{Q_n(e^n|y^n)}
= \mathcal O(\log n).
\label{eq:ggtiltgridregret}
\end{equation}
Therefore, for every $y^n\in\mathcal T_n$, $m_n^{\mathcal V_{\widetilde\Theta}}(e^n|y^n) \le 2^{\Gamma_n}Q_n(e^n|y^n)$. Summing over $e^n\in\{0,1\}^n$ gives
\begin{equation}
C_n^{\mathcal V_{\widetilde\Theta}}(y^n) \coloneqq \sum_{e^n\in\{0,1\}^n} m_n^{\mathcal V_{\widetilde\Theta}}(e^n|y^n)
\le 2^{\Gamma_n},
\end{equation}
and hence
\begin{equation}
\sup_{y^n\in\mathcal T_n} \log_2 C_n^{\mathcal V_{\widetilde\Theta}}(y^n)
= \mathcal O(\log n) = o(n).
\label{eq:ggtiltShtarkov}
\end{equation}
Because $\theta\in\Theta_0\subset\widetilde\Theta$, \eqref{eq:ggtiltbadset} also yields
\[
\limsup_{n\to\infty} \frac{1}{n} \log_2 P_\theta\{Y^n\notin\mathcal T_n\} \le -\kappa.
\]
Theorem~\ref{thm:tiltclosuremoments}, with $\bar\rho=1$, therefore gives \eqref{eq:ggtiltcompletemoment}. It remains to verify the differentiability and nondegeneracy conditions of Theorem~\ref{thm:tiltclosureerror}. For $\rho\ge0$, set $\alpha=\frac{1}{1+\rho}$ and define
\begin{equation}
A_\theta(\rho) \coloneqq \int_{\mathbb R}
\left[ \left( \frac{f_{\beta,\sigma}(y-1)}{2} \right)^\alpha
+ \left( \frac{f_{\beta,\sigma}(y+1)}{2} \right)^\alpha \right]^{1/\alpha} \mathrm dy.
\label{eq:ggArimotoA}
\end{equation}
For each fixed $y$, the map $x\mapsto x\oplus h(y)$ is a permutation of $\{0,1\}$. Hence, for $\rho>0$, $H_{\alpha,\theta}^{\rm A}(E|Y) = H_{\alpha,\theta}^{\rm A}(X|Y)$, and the definition of the matched spectrum gives
\begin{equation}
\Lambda_{{\rm m},\theta}(\rho)
=
\log_2 A_\theta(\rho).
\label{eq:ggLambdaA}
\end{equation}
The same identity holds at $\rho=0$, because $A_\theta(0)=1$ and $\Lambda_{{\rm m},\theta}(0)=0$.
Fix $K>0$. For $\rho\in[0,K]$, $\alpha\in[(1+K)^{-1},1]$. Let $a(y)=\frac12 f_{\beta,\sigma}(y-1)$, $b(y)=\frac12 f_{\beta,\sigma}(y+1)$, and define $S_\rho(y)=a(y)^\alpha+b(y)^\alpha$, $g_\rho(y)=S_\rho(y)^{1/\alpha}$. By concavity of $u\mapsto u^\alpha$ on $[0,\infty)$,
\begin{equation}
g_\rho(y) \le 2^{1/\alpha-1}\bigl(a(y)+b(y)\bigr)
\le 2^K\bigl(a(y)+b(y)\bigr).
\label{eq:ggintegrandbound}
\end{equation}
Because $a(y),b(y)>0$, $g_\rho(y)$ is continuously differentiable in $\rho$. Define
\[
w_a(y) = \frac{a(y)^\alpha}{S_\rho(y)}, \qquad w_b(y) = \frac{b(y)^\alpha}{S_\rho(y)}.
\]
Since $\partial_\rho\alpha=-\alpha^2$,
\begin{equation}
\partial_\rho\ln g_\rho(y) = \ln S_\rho(y) - \alpha
\left[ w_a(y)\ln a(y) + w_b(y)\ln b(y) \right].
\label{eq:gglogderivative}
\end{equation}
Let $m(y)=\max\{a(y),b(y)\}$. Then
\[
m(y)^\alpha \le S_\rho(y) \le 2m(y)^\alpha.
\]
Consequently,
\begin{equation}
\left|
\partial_\rho g_\rho(y)
\right|
\le
g_\rho(y)
\left[
\ln2
+
2|\ln a(y)|
+
2|\ln b(y)|
\right].
\label{eq:ggderivativebound}
\end{equation}
For the generalized-Gaussian density, there exist finite constants $C_\theta,C_\theta',c_\theta>0$ such that, for all $y\in\mathbb R$,
\begin{align}
|\ln a(y)|+|\ln b(y)| &\le C_\theta(1+|y|^\beta),
\label{eq:gglogdensitybound}\\
a(y)+b(y) &\le C_\theta' \exp(-c_\theta|y|^\beta).
\label{eq:ggdensitytailbound}
\end{align}
Equations~\eqref{eq:ggintegrandbound}, \eqref{eq:ggderivativebound}, \eqref{eq:gglogdensitybound}, and \eqref{eq:ggdensitytailbound} provide an integrable dominating function, independent of $\rho\in[0,K]$, for both $g_\rho$ and $\partial_\rho g_\rho$. Differentiation under the integral is therefore valid, and dominated convergence applied to $\partial_\rho g_\rho$ shows that the derivative is continuous on $[0,K]$. Thus $A_\theta\in C^1([0,K])$. Since $A_\theta(\rho)>0$, \eqref{eq:ggLambdaA} yields $\Lambda_{{\rm m},\theta}\in C^1([0,K])$. As $K>0$ was arbitrary, $\Lambda_{{\rm m},\theta}$ is continuously differentiable on every compact interval in $[0,\infty)$.
Finally, $f_{\beta,\sigma}(z)>0$ for every $z\in\mathbb R$, so $0<P_\theta(X=0|Y=y)<1$ for every $y$, which implies $H_\theta(X|Y)>0$. Moreover, $P_\theta(X=0|Y=y)=\frac12$ if and only if $f_{\beta,\sigma}(y-1)=f_{\beta,\sigma}(y+1)$. By the generalized-Gaussian density formula, this equality is equivalent to $|y-1|=|y+1|$, and hence to $y=0$. Since $Y$ has a density, $P_\theta\{Y=0\}=0$. Therefore $h_2\!\left(P_\theta(X=0|Y)\right)<1$ almost surely, where $h_2$ is the binary entropy function, and hence $H_\theta(X|Y)<1$. Finally, conditional on $Y$, the map $X\mapsto E=X\oplus h(Y)$ is bijective, so $H_\theta(E|Y)=H_\theta(X|Y)$. All hypotheses of Theorem~\ref{thm:tiltclosureerror} are therefore satisfied, and \eqref{eq:ggtiltcompleteerror} follows.
\end{proof}

Parameter values in $\widetilde\Theta\setminus\Theta_0$ are auxiliary decoder parameters and do not belong to the prescribed channel parameter set $\Theta_0$.

\section{Finite-Block Numerical Illustrations and Gallager Benchmark}
\label{sec:numerical}

The numerical studies examine finite-block order reversal and posterior-normalized rank regret, compare the matched order with an auxiliary posterior-envelope order containing the conditional power tilts required by the sufficient exponent-preservation condition for $\rho\in[0,1]$, and estimate the channel-averaged uniform-subset error probability relative to Gallager's uniform-input random-coding exponent.

Figure~\ref{fig:ggn_numerical}(a) plots the generalized-Gaussian correction costs from Section~\ref{sec:signalexamples}. For $e_A=(1,0,0)$ and $e_B=(0,1,1)$ at received magnitudes $(2,0.6,0.6)$, define $c_\beta(e_A)=\ell_\beta(2)$, $c_\beta(e_B)=2\ell_\beta(0.6)$. Numerically, the costs are equal at $\beta\approx1.257$. Since $c_1(e_A)<c_1(e_B)$ and $c_2(e_A)>c_2(e_B)$, the two corrections have opposite relative orders at $\beta=1$ and $\beta=2$.

For Fig.~\ref{fig:ggn_numerical}(b), let $\Theta_{\rm f} = \{1,1.25,1.5,1.75,2\} \times \{0.8,1,1.2\}$ be a set of $15$ generalized-Gaussian parameter pairs. For $n\in\{4,8,12,16\}$, all $2^n$ correction patterns are enumerated for the magnitude block obtained by truncating or periodically repeating $(3.92,0.86,1.86,0.90,0.73,1.12,1.37,0.76)$. Let $r^{(n)}$ denote the resulting length-$n$ magnitude block. Under the sign reference decision, the generalized-Gaussian correction posterior depends on $y_i$ only through $|y_i|$; throughout the deterministic-block calculations below, $p_{\theta,n}(\cdot|r^{(n)})$ denotes this posterior evaluated at any received block with magnitudes $r^{(n)}$. All such posterior probabilities are strictly positive. For an ordering $\pi_n$, define
\begin{align}
\mathfrak r_n \bigl( \pi_n,\Theta_{\rm f}\mid r^{(n)} \bigr)
&= \max_{\theta\in\Theta_{\rm f}} \max_{e^n\in\{0,1\}^n} \log_2 \left[ G_{\pi_n} \bigl( e^n\mid r^{(n)} \bigr) p_{\theta,n} \bigl( e^n\mid r^{(n)} \bigr) \right],
\label{eq:numpointregret}\\
\mathfrak r_n^\star \bigl( \Theta_{\rm f}\mid r^{(n)} \bigr)
&= \min_{\pi_n(r^{(n)})} \mathfrak r_n \bigl( \pi_n,\Theta_{\rm f}\mid r^{(n)} \bigr),
\label{eq:numpointminimax}
\end{align}
and
\begin{equation}
\Delta \mathfrak r_n(\pi_n) = \mathfrak r_n \bigl( \pi_n,\Theta_{\rm f}\mid r^{(n)} \bigr)
- \mathfrak r_n^\star \bigl( \Theta_{\rm f}\mid r^{(n)} \bigr).
\label{eq:numexcessregret}
\end{equation}
By Theorem~\ref{thm:exactregret}, the posterior-envelope order satisfies $\Delta \mathfrak r_n=0$. For each displayed block, the computed equal-weight posterior-mixture order over $\Theta_{\rm f}$ also attains $\Delta \mathfrak r_n=0$. Each fixed plug-in order shown in Fig.~\ref{fig:ggn_numerical}(b) has $\Delta \mathfrak r_n>0$ at every displayed blocklength.

Exact ties arising from repeated magnitudes are preserved and resolved only by the prescribed deterministic pattern tie-breaking rule. Corrections having identical ordering scores are therefore treated as members of the same score-tie class. For $\beta\in \{1,2\}$, the identities $\ell_1(r)=2\min \{r,1\}$ and $\ell_2(r)=4r$ permit exact rational evaluation of the corresponding subset-sum scores for the decimal-valued observation blocks, so no numerical tolerance is required to determine score equality. This posterior-mixture equality holds for the displayed observation blocks and the finite parameter set $\Theta_{\rm f}$; no approximation result for the continuous parameter family follows.

\begin{figure*}[t]
\centering
\subfloat[\label{fig:ggn-order-reversal}]{%
    \includegraphics[width=0.45\textwidth]{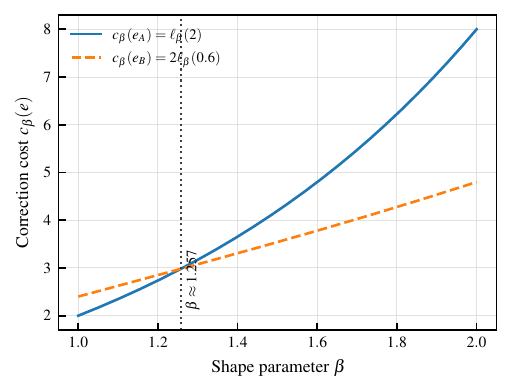}%
}
~
\subfloat[\label{fig:ggn-finite-regret}]{%
    \includegraphics[width=0.45\textwidth]{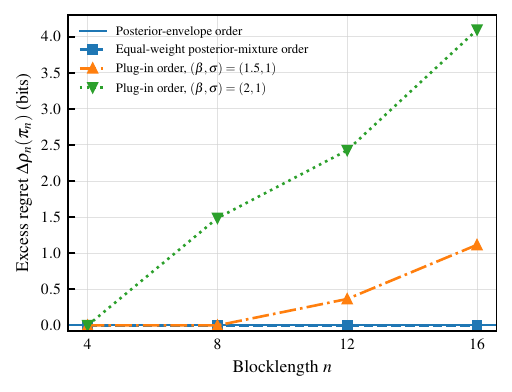}%
}
\vspace{-7pt}
\caption{Finite-block generalized-Gaussian illustrations.
(a) Correction costs for $e_A$ and $e_B$.
(b) Excess posterior-normalized rank regret relative to the finite-class minimax value for $\Theta_{\rm f}$.}
\label{fig:ggn_numerical}
\vspace{-5pt}
\end{figure*}

For a finite-block illustration of the tilt-containing construction, fix $\theta_0=(\beta,\sigma)=(1.5,1)$ and define the auxiliary posterior family
\begin{equation}
\mathcal V_{\rm tc} = \left\{ p_{\beta,\sigma'} \,\middle|\,
\beta\in\{1,1.25,1.5,1.75,2\},
\quad 0.8\le\sigma'\le2.4 \right\}.
\label{eq:numtiltclosedfamily}
\end{equation}
For $\theta_0$, $\sigma_\rho=(1+\rho)^{2/3}$, $0\le\rho\le1$, so the full conditional power-tilt path for $\rho\in[0,1]$ is contained in $\mathcal V_{\rm tc}$.

For $n\in\{4,6,8,10,12\}$, the magnitude block is obtained by truncating or periodically repeating $(3.92,0.86,$ $1.86,0.90,0.73,1.12,1.37,0.76)$. All $2^n$ corrections are enumerated, and $M_n=2^{n/2}$. For each listed $\beta$, the supremum over $\sigma'\in[0.8,2.4]$ is evaluated separately. Equivalently, $t=(\sigma')^{-\beta} \in[2.4^{-\beta},0.8^{-\beta}]$. For fixed $e^n$ and $r^n$,
\begin{align}
\ln p_{\beta,\sigma'}(e^n|r^n) &= t T_\beta(e^n;r^n)
- \sum_{i=1}^n \ln\!\left( 1+\exp[t\ell_\beta(r_i)] \right),
\label{eq:numscaleobjective}\\
T_\beta(e^n;r^n) &\coloneqq \sum_{i:e_i=0}\ell_\beta(r_i).
\end{align}
Moreover,
\begin{equation}
\frac{\partial^2}{\partial t^2} \ln p_{\beta,\sigma'}(e^n|r^n) =
- \sum_{i=1}^n \ell_\beta(r_i)^2 \frac{\exp[t\ell_\beta(r_i)]}{\left(1+\exp[t\ell_\beta(r_i)]\right)^2} \le0.
\end{equation}
Thus the conditional log-posterior is concave in $t$. Its maximum over the compact scale interval is attained either at an endpoint or at an interior stationary point. For each fixed $\beta$, corrections with the same $T_\beta(e^n;r^n)$ have the same scale-optimized posterior value. Exact ties are retained before applying the prescribed deterministic tie-breaking rule.

For a deterministic order $\pi_n$, define the normalized logarithm of the conditional mean rank
\begin{equation}
\mu_n(\pi_n|r^n) = \frac1n \log_2 \sum_{e^n} p_{\theta_0,n}(e^n|r^n) G_{\pi_n}(e^n|r^n),
\label{eq:numconditionalmoment}
\end{equation}
and, for codebook size $M$, the conditional uniform-subset ensemble error probability
\begin{align}
\overline P_{e,\pi_n}^{(n)}(r^n;M) &= \sum_{e^n} p_{\theta_0,n}(e^n|r^n)
\left[ 1- \frac{ \binom{ 2^n-G_{\pi_n}(e^n|r^n)}{M-1}}{\binom{2^n-1}{M-1}} \right].
\label{eq:numconditionalerror}
\end{align}
with the convention that $\binom{m}{k}=0$ for $m<k$.

Figure~\ref{fig:ggn_tilt_finite} compares the matched posterior order with the envelope order induced by \eqref{eq:numtiltclosedfamily}. In \eqref{eq:numconditionalerror}, this study uses $M=M_n=2^{n/2}$. The displayed conditional mean-rank and ensemble-error quantities differ at finite blocklength. These finite-block computations use an auxiliary family containing the full conditional power-tilt path for $\rho\in[0,1]$ at $\theta_0$; no asymptotic exponent equality or blocklength monotonicity follows from them.

\begin{figure*}[t]
\centering
\subfloat[\label{fig:ggn-tilt-rank}]{%
    \includegraphics[width=0.45\textwidth]{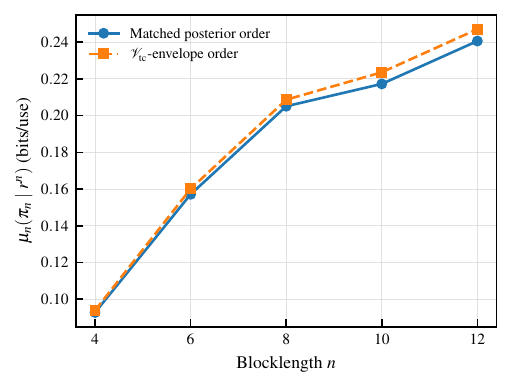}%
}
~
\subfloat[\label{fig:ggn-tilt-error}]{%
    \includegraphics[width=0.45\textwidth]{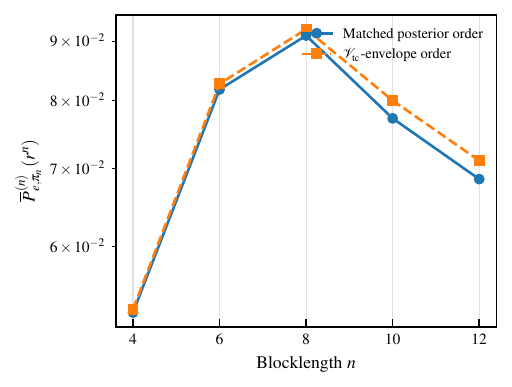}%
}
\vspace{-7pt}
\caption{Finite-block matched and tilt-containing envelope orders for
$(\beta,\sigma)=(1.5,1)$.
(a) Normalized logarithm of the conditional mean rank.
(b) Conditional uniform-subset ensemble error probability at $R_{\rm c}=1/2$.}
\label{fig:ggn_tilt_finite}
\vspace{-5pt}
\end{figure*}

\subsection{Channel-averaged finite-block error estimates and Gallager benchmark}
\label{subsec:numgallager}

Fix $\theta_0=(\beta,\sigma)=(1.5,1)$, $R_{\rm c}=\frac14$, and let $f=f_{1.5,1}$ and $\alpha_\rho=(1+\rho)^{-1}$. The matched rank-moment spectrum is $\Lambda_{\rm m}(\rho) = \log_2 \int_{\mathbb R} \frac12 \left[ f(y-1)^{\alpha_\rho} + f(y+1)^{\alpha_\rho} \right]^{1/\alpha_\rho} \, \mathrm{d}y$. The integral is truncated to $[-16,16]$ and evaluated using a $600$-node Gauss--Legendre rule separately on $[-16,-1]$, $[-1,1]$, and $[1,16]$. As a numerical stability check, repeating the computation with $300$ and $600$ nodes and truncation limits between $12$ and $20$ changes the maximized exponent by less than $2\times10^{-11}$ bits/use. A separate numerical evaluation of $E_0^{\rm U}(\rho)$ from its defining integral agrees with $E_0^{\rm U}(\rho)=\rho-\Lambda_{\rm m}(\rho)$ to the reported numerical precision. Numerical maximization over $0\le\rho\le1$ gives $E_r^{\rm U}(1/4)=0.14448327\ \text{bits/use}$, $\rho^\star\approx0.99882$. For the envelope order, use
\[
\mathcal V_{\rm exp} = \left\{ p_{\beta,\sigma'} \,\middle|\, \beta\in\{1,1.5,2\},
\quad 0.8\le\sigma'\le2.4 \right\}.
\]
For $\theta_0$, the conditional power-tilt scales satisfy
\[
\sigma_\rho=(1+\rho)^{2/3} \in[1,2^{2/3}],
\qquad 0\le\rho\le1,
\]
so the full conditional power-tilt path for $\rho\in[0,1]$ is contained in $\mathcal V_{\rm exp}$.

For $n\in\{4,8,12\}$, set $M_n=2^{n/4}$. The supremum over the scale parameter in the envelope is evaluated using the concave $t$-optimization in \eqref{eq:numscaleobjective}. For an ordering $\pi_n$, define
\begin{equation}
\overline P_{e,\pi_n}^{(n)} = \E_{\theta_0} \left[ \overline P_{e,\pi_n}^{(n)}(|Y^n|;M_n) \right],
\label{eq:numchannelaverage}
\end{equation}
where $|Y^n|=(|Y_1|,\ldots,|Y_n|)$ and $\overline P_{e,\pi_n}^{(n)}(r^n;M_n)$ is given by \eqref{eq:numconditionalerror}. The expectation is estimated from $N_{{\rm MC},4}=2000$, $N_{{\rm MC},8}=1000$, and $N_{{\rm MC},12}=300$ independent output blocks. For each sampled block, all $2^n$ corrections are enumerated and the conditional error probability is computed from the exact hypergeometric expression. Monte Carlo sampling is therefore used only for the expectation over $Y^n$.

\begin{figure*}[t]
\centering
\subfloat[\label{fig:ggn-gallager-functions}]{%
    \includegraphics[width=0.445\textwidth]{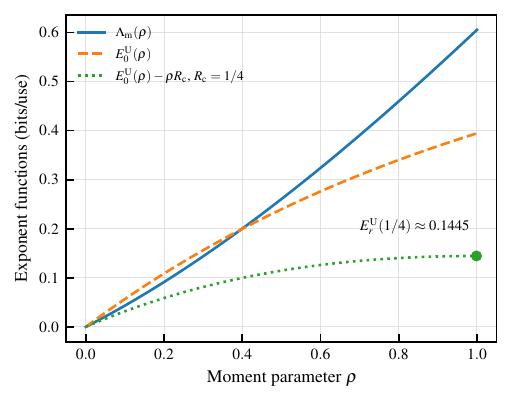}%
}
~
\subfloat[\label{fig:ggn-channel-averaged}]{%
    \includegraphics[width=0.445\textwidth]{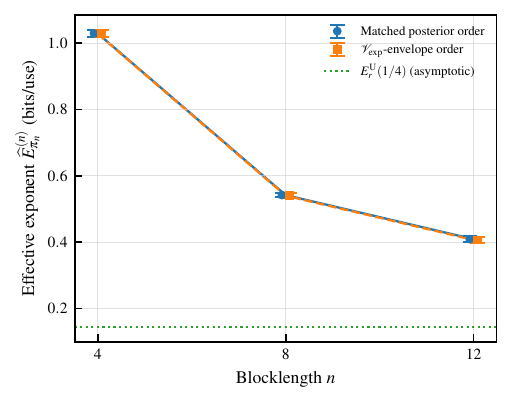}%
}
\vspace{-7pt}
\caption{Gallager benchmark and channel-averaged finite-block comparison for generalized-Gaussian BPSK with $(\beta,\sigma)=(1.5,1)$ and $R_{\rm c}=1/4$.
(a) $\Lambda_{\rm m}(\rho)$, $E_0^{\rm U}(\rho)$, and $E_0^{\rm U}(\rho)-\rho R_{\rm c}$.
(b) Matched and envelope effective-exponent estimates; error bars are monotone transforms of pointwise nominal $95\%$ Student-$t$ intervals for the mean conditional error probability over channel outputs.}
\label{fig:ggn_gallager}
\end{figure*}

Let
\begin{equation}
\widehat P_{e,\pi_n}^{(n)} = \frac1{N_{{\rm MC},n}} \sum_{j=1}^{N_{{\rm MC},n}} \overline P_{e,\pi_n}^{(n)}(|Y_j^n|;M_n)
\end{equation}
and define the corresponding Monte Carlo plug-in effective-exponent estimate
\begin{equation}
\widehat E_{\pi_n}^{(n)} = -\frac1n \log_2 \widehat P_{e,\pi_n}^{(n)}.
\label{eq:numeffectiveexponent}
\end{equation}
Figure~\ref{fig:ggn_gallager}(b) reports \eqref{eq:numeffectiveexponent} for the matched and envelope orders. The markers and error bars are horizontally offset by $\mp0.08$ for visibility; the connecting curves use the actual blocklengths.

All Monte Carlo experiments use the fixed random seed $20260819$, with separate pseudorandom streams across blocklengths. For each order and blocklength, let $[L_n,U_n]$ be the pointwise nominal $95\%$ Student-$t$ interval for the mean in \eqref{eq:numchannelaverage}. Whenever $L_n>0$, $p\mapsto-n^{-1}\log_2p$ is strictly decreasing on $(0,\infty)$, so the monotone transform of this interval is $\left[ -\frac1n\log_2 U_n,\, -\frac1n\log_2 L_n \right]$. These transformed intervals quantify only Monte Carlo uncertainty in the estimated channel-output expectation; they do not include numerical quadrature or scale-optimization error. They are nominal Student-$t$ intervals for a bounded, potentially skewed conditional-error variable and are not interpreted as having exact finite-sample $95\%$ coverage. Because the conditional-error variable is bounded but potentially skewed, these are nominal Student-$t$ intervals rather than intervals with exact finite-sample $95\%$ coverage. They quantify Monte Carlo uncertainty in the estimated channel-output expectation. The smallest displayed effective-exponent lower endpoint is $0.39828$ bits/use, attained at $n=12$ for the envelope order. No convergence conclusion follows from its separation from $E_r^{\rm U}(1/4)$ at this blocklength. Writing $\widehat P_{e,{\rm m}}^{(n)}$ and $\widehat P_{e,{\rm env}}^{(n)}$ for the matched and envelope point estimates, respectively, $100\left( \frac{\widehat P_{e,{\rm env}}^{(n)}}{\widehat P_{e,{\rm m}}^{(n)}} -1 \right)\%$ equals $0.294\%$, $0.930\%$, and $1.751\%$ at $n=4$, $8$, and $12$, respectively. The three finite-block values do not establish monotonicity with blocklength. Under the full hypotheses of Theorem~\ref{thm:tiltclosureerror}, the corresponding asymptotic exponent equality holds.

\section{Conclusion}

Under channel uncertainty, admissible channels may induce different matched correction orders for the same observation. For each admissible observation, the posterior-envelope order minimizes the worst-case posterior-normalized rank regret over admissible channel--correction pairs with positive posterior probability. Under the uniform-subset code ensemble, this pointwise minimax result does not by itself determine the ensemble-average error exponent, which depends on the distribution of the realized correction rank, including its right tail. For a fixed memoryless channel under uniform input, an auxiliary posterior-envelope order preserves the matched rank-moment spectrum on $[0,1]$ when the decoder posterior family is $1$-tilt closed at the channel parameter and has subexponential conditional Shtarkov complexity on observation sets whose complement probabilities decay exponentially faster than the reciprocal correction-space cardinality. If the matched spectrum is continuously differentiable on $[0,1+\epsilon]$ for some $\epsilon>0$ and $0<H_\theta(E|R)<1$, then, under the uniform-subset code ensemble, for every fixed $0<R_{\rm c}<1$ and every code-size sequence satisfying $\log_a M_n=nR_{\rm c}+o(n)$, the auxiliary order attains the matched ensemble-average error exponent, which equals Gallager's uniform-input random-coding exponent. For generalized-Gaussian BPSK with uncertain shape and scale, the posterior-grid first-order conditions and, separately, the tilt-containing product-envelope exponent-preservation conditions hold. Changing the scale at fixed shape leaves the correction order unchanged, whereas changing the shape can reverse the relative order of two correction patterns for the same received block. Conditional power tilting preserves the shape parameter and changes only the scale parameter. The prescribed channel class determines the finite-block minimax envelope, whereas the decoder posterior family determines the product-posterior envelope used for exponent preservation. The decoder posterior family may coincide with the posterior family induced by the prescribed channel class, or it may include auxiliary posteriors not induced by channels in that class.

\section*{Acknowledgment}

This work was supported by the Swiss National Science Foundation (SNSF)
under Grant No.~222339.

\bibliographystyle{IEEEtran}
\bibliography{references}

\end{document}

%% file: preamble_arXiv.tex
\usepackage[letterpaper,margin=1in]{geometry}
\usepackage[utf8]{inputenc}
\usepackage[T1]{fontenc}
\usepackage{microtype}

\usepackage{amsmath}
\usepackage{amssymb}
\usepackage{amsthm}
\usepackage{mathtools}

\usepackage{graphicx}
\usepackage{subfig}
\DeclareGraphicsExtensions{.pdf,.png,.jpg,.jpeg}

\usepackage{enumitem}
\usepackage{cite}
\usepackage{url}

\usepackage[
    colorlinks=true,
    linkcolor=red,
    citecolor=blue,
    urlcolor=blue
]{hyperref}
\usepackage{orcidlink}

\providecommand{\keywords}[1]{%
  \par\vspace{0.5em}%
  \noindent\textbf{\textit{Keywords---}} #1\par
}

\theoremstyle{plain}
\newtheorem{theorem}{Theorem}
\newtheorem{lemma}{Lemma}
\newtheorem{proposition}{Proposition}
\newtheorem{corollary}{Corollary}

\theoremstyle{definition}
\newtheorem{definition}{Definition}
\newtheorem{assumption}{Assumption}
\newtheorem{example}{Example}

\theoremstyle{remark}
\newtheorem{remark}{Remark}

\newcommand{\A}{\mathcal{A}}
\providecommand{\C}{}
\renewcommand{\C}{\mathcal{C}}
\newcommand{\Wcal}{\mathcal{W}}
\newcommand{\Scal}{\mathcal{S}}
\newcommand{\Rcal}{\mathcal{R}}

\newcommand{\E}{\mathbb{E}}
\newcommand{\Prb}{\mathbb{P}}

\newcommand{\oplusA}{\oplus}
\newcommand{\ominusA}{\ominus}

\newcommand{\ceil}[1]{\left\lceil #1\right\rceil}

\DeclareMathOperator*{\argmax}{arg\,max}
\newcommand{\plimsup}{\operatorname*{p-lim\,sup}}
\newcommand{\pliminf}{\operatorname*{p-lim\,inf}}
\newcommand{\esssup}{\operatorname*{ess\,sup}}

\newcommand{\id}{\mathbf{1}}
\newcommand{\Pe}{P_{\mathrm e}}